\documentclass[12pt,draftclsnofoot,onecolumn]{IEEEtran}
\IEEEoverridecommandlockouts
\pdfoutput=1
\usepackage{mathtools}

\usepackage{lipsum}

\usepackage{amsthm}
\usepackage{cite}
\usepackage{amsmath,amssymb,amsfonts}
\usepackage{algorithmic}
\usepackage{graphicx}
\usepackage{textcomp}
\usepackage{xcolor}
\usepackage{float}
\def\BibTeX{{\rm B\kern-.05em{\sc i\kern-.025em b}\kern-.08em
    T\kern-.1667em\lower.7ex\hbox{E}\kern-.125emX}}
\usepackage{amsmath}
\usepackage{csquotes}
\usepackage[linesnumbered, ruled]{algorithm2e}

\usepackage{caption}
\usepackage{subcaption}
\usepackage[utf8]{inputenc}
\usepackage[english]{babel}
\usepackage{dsfont}
\usepackage{amsmath}
\makeatletter
\newcommand{\mathleft}{\@fleqntrue\@mathmargin0pt}
\newcommand{\mathcenter}{\@fleqnfalse}
\makeatother
\usepackage{csquotes}
\usepackage[utf8]{inputenc}
\usepackage[english]{babel}
\usepackage{dsfont}

\newtheorem{lemma}{Lemma}
\newtheorem{remark}{Remark}

\usepackage{bbm}

\usepackage{amsmath}

\DeclareMathOperator*{\argmin}{arg\,min}
\DeclareMathOperator*{\limisup}{lim\,sup}
\DeclareMathOperator*{\amin}{\mbox{minimize}}
\usepackage[linesnumbered,ruled]{algorithm2e}
\usepackage[nodisplayskipstretch]{setspace}
\begin{document}
\title{
Multi-Source AoI-Constrained Resource Minimization under HARQ: Heterogeneous Sampling Processes
}
\author{Saeid Sadeghi Vilni$^*$, Mohammad Moltafet$^*$, Markus Leinonen$^*$, and Marian~Codreanu$^\dagger$
\thanks{$^*$Centre for Wireless Communications–Radio Technologies, University of Oulu, 90014 Oulu, Finland (e-mail: firstname.lastname@oulu.fi). $^\dagger$Department of Science and Technology, Linkoping University, Sweden (e-mail: marian.codreanu@liu.se).}
}
\date{}
\maketitle

\vspace{-10mm}
\begin{spacing}{1.45}
\begin{abstract}
We consider a multi-source hybrid automatic repeat request (HARQ) based system, where a transmitter sends status update packets of \textit{random arrival} (i.e., uncontrollable sampling) and \textit{generate-at-will} (i.e., controllable sampling) sources to a destination through an error-prone channel. We develop transmission scheduling policies to minimize the average number of transmissions subject to an average age of information (AoI) constraint. First, we consider known environment (i.e., known system statistics) and develop a near-optimal deterministic transmission policy and a low-complexity dynamic transmission (LC-DT) policy. The former policy is derived by casting the main problem into a constrained Markov decision process (CMDP) problem, which is then solved using the Lagrangian relaxation, relative value iteration algorithm, and bisection. The LC-DT policy is developed via the drift-plus-penalty (DPP) method by transforming the main problem into a sequence of per-slot problems. Finally, we consider unknown environment and devise a learning-based transmission policy by relaxing the CMDP problem into an MDP problem using the DPP method and then adopting the deep Q-learning algorithm. Numerical results show that the proposed policies achieve near-optimal performance and illustrate the benefits of HARQ in status updating.   


\end{abstract}
\end{spacing}\vspace{-8mm}








\sloppy

%
\IEEEpeerreviewmaketitle

\section{Introduction}
There is a growing demand for Internet-of-Things
(IoT) and cyber-physical systems such as autonomous vehicles, wireless industrial automation, and health monitoring that rely heavily on real-time (fresh) status updates. In these systems, a source  (containing a sensor) monitors a physical phenomenon such as temperature, pressure, or motion and sends status updates to a destination (e.g., a remote controller) for decision-making \cite{aoi1,oiot}. The Age of Information (AoI) \cite{aoi1,aoi2,oiot} is a metric used to evaluate the freshness of information in the status update systems. AoI is defined as the difference between the current time and the generation time of the last received packet at a destination \cite{aoi1,aoi2,oiot}.
Each status update packet contains a timestamp representing the time when the sample was generated and the measured value of the monitored process. At time instant $t$, denoting the timestamp of the last received status update packet by $U_t$, the AoI, $\delta_t$, is defined as $\delta_t = t-U_t$ \cite{aoi1,oiot, aoi2,ry,pp}.

The reliability of data transmissions under an unreliable communication channel can be enhanced via retransmission protocols \cite{arqb}. Automatic repeat request (ARQ) protocols are standard error control methods, where
after each transmission, the transmitter receives a feedback about the reception status of the packet as acknowledgement/negative-acknowledgement (ACK/NACK) \cite{arqb}. The transmitter keeps retransmitting each packet until it receives an ACK or reaches the maximum allowed number of retransmissions. The ARQ protocols use only the last received version of a packet for decoding, whereas the hybrid ARQ (HARQ) protocols use all received versions, thus increasing the probability of successfully decoding the packet \cite{arqb,harq}.

In this paper, we consider a multi-source HARQ-based status update system, where the sources are connected to a transmitter that sends status update packets to a receiver over an unreliable wireless channel (see Fig.~\ref{s1}). We assume a slotted communication, in which the transmitter can send at most one packet per slot. The sources, which monitor some time-varying random processes, are classified into two categories based on their sampling processes: 1) \textit{random arrival} sources (i.e., uncontrollable sampling) which generate status update packets according to a Bernoulli process, and 2) \textit{generate-at-will} sources (i.e., controllable sampling) which can be commanded to generate a status update packets at any slot. Our considered system may represent a multi-source node such as a multi-sensor IoT device that is equipped with a single transmitter to communicate all the different sensed data to a remote location through a wireless channel. In such scenario, the transmitter has control over the sampling of some sources, e.g., on-demand requesting of samples from a temperature or moisture sensor. On the contrary, the sampling of some sources depends on other grounds, such as energy to generate a sample (e.g., the source needs to harvest energy) or time to generate a sample (e.g., the source needs to scan an area which takes a random amount of time), leading to generating packets at random times. We further consider that each random arrival source has a buffer to retain the last (randomly) generated packet.
Furthermore, in order to benefit from the HARQ, the transmitter has a memory to store the last transmitted but not successfully decoded packet of each source as this packet has a higher chance of being decoded than a new packet.

Apart from freshness requirements, the radio resources (e.g., power and channel utilization) also play an essential role in the operation of status update systems \cite{eiage}. 
Hence, we investigate the problem of minimizing the average number of transmissions subject to the average AoI constraint. The solution of the problem determines the transmission status at each time slot, i.e., transmit a fresh packet from a source, retransmit the previously transmitted but not successfully decoded packet from a source, or stay idle.

A scenario where the controller knows the probability of possible outcomes after making a decision in a system is called a known environment \cite{env}. In our considered system, the known environment corresponds to the case where the transmitter knows the packet arrival rate of the random arrival sources and the probability of successful decoding after each transmission attempt.
Since, in some cases, the known environment is not accessible, we investigate the problem in both the known environment and the unknown environment.
We propose three solutions to the problem, namely, a (stationary) deterministic transmission policy and a low-complexity dynamic transmission (LC-DT) policy for the known environment and a learning-based transmission policy for the unknown environment.

To obtain the deterministic transmission policy in the known environment, we cast the main problem as a constrained Markov decision process (CMDP) problem. Then, we transform the CMDP problem into an MDP problem via the Lagrangian relaxation. In general, an optimal policy for the CMDP problem is a randomized mixture of two deterministic policies, where one deterministic policy is feasible (satisfies the constraint) and the another policy is infeasible \cite{beutler}; see also recent applications \cite{bz,ngz,gz}. 
However, since obtaining such randomized policy is often computationally intractable, we propose a near-optimal {practical} deterministic transmission policy (feasible deterministic policy) and a lower-bound policy (infeasible deterministic policy) {for benchmarking purposes} using relative value iteration algorithm (RVIA) and the bisection algorithm. Since the number of states to explore in RVIA increases exponentially in the number of sources and RVIA is run at each bisection iteration, obtaining the deterministic transmission policy is inefficient computationally. Therefore, we propose the LC-DT policy by using the drift-plus-penalty (DPP) method \cite{lyp}. According to the DPP method, the average AoI constraint is transformed into a queue stability constraint, and subsequently, the time average main problem is transformed into an optimization problem that is to be solved at each time slot.
To obtain the learning-based transmission policy, we use the DPP method to transform the CMDP problem into an MDP problem, in which we minimize the time average DPP function. Then, we develop the policy by solving the MDP problem with a deep Q-learning (DQL) algorithm \cite{dqn}.

In the numerical results, we analyze the effectiveness of the proposed policies. We compare the effectiveness and complexity of the policies and study the effect of employing HARQ. 


The main contributions of the paper are summarized as
follows:
\begin{itemize}
    \item We consider a multi-source HARQ-based status update system that consists of random arrival and generate-at-will sources. We minimize the average number of transmissions under the average AoI constraint in the known and unknown environment.
    \item For the known environment, we develop a deterministic transmission policy using the Lagrangian relaxation, RVIA, and the bisection. Moreover, we propose a low-complexity dynamic transmission policy using the DPP method.
    \item For the unknown environment, we develop a learning-based transmission policy by using the DPP method to cast the main problem as an MDP problem, which is then solved by applying DQL.
    \item The numerical results demonstrate the near-optimal performance of the proposed transmission policies compared to the lower-bound policy, and a significant improvement respect to a baseline policy. The results corroborate that HARQ improves the performance of the system and illustrate that the learning-based transmission policy performs close to the policies developed for the known environment. 
\end{itemize}

\subsection{Related Work}
AoI characterization has extensively been studied from the perspective of queueing theory; see, e.g., \cite{ry2,mm1,enj2,nr,nnr} and the references therein. One of the earliest studies to analyze AoI under an HARQ protocol is \cite{enj2}, where the authors derived the closed-form expression of the average AoI for an HARQ-based M/G/1/1 queueing system.
Considering the queueing system and the derived AoI result from \cite{enj2}, the work \cite{nr} studied the age-optimal redundancy allocation problem under a constraint on the decoding error probability for both chase combining and incremental redundancy HARQ protocols.
An {M/M/1} queueing system with network-code-HARQ protocol is considered in \cite{nnr}, where the closed-form expression of AoI is derived.

Besides the analysis, the AoI has been studied in the retransmission-based status update systems from the perspective of sampling and transmission policies
\cite{frzi,shi2021,arafa2021,deng2021,feng2021,wang2020,gz,ngz}.
In \cite{frzi}, the authors considered a multi-source and generate-at-will-based status update system and minimized the average AoI by proposing a source selection policy under three pre-defined transmission policies.
In \cite{shi2021}, the authors derived the closed-form expression of the average AoI in an HARQ-based status update system in which two energy harvesting sources send the same information for providing diversity at the destination.
The work \cite{arafa2021} considered a multi-source status update system in which the transmitter harvests energy and uses a greedy retransmission policy. They minimized the average AoI by determining a set of transmission times and choosing a source to send status update packet.
An HARQ-based non-orthogonal multiple access system with two users are considered in \cite{deng2021}, where the average AoI is  minimized by determining the transmit power and transmission status, i.e., transmitting a new packet or retransmitting the previously transmitted but not successfully decoded packet, at each slot.
In \cite{feng2021}, the authors investigated the average AoI minimization problem in a status update system with a pre-defined retransmission policy to find the times for updating the destination.
The work \cite{wang2020} studied the average AoI minimization problem in an HARQ-based status update system. They calculated the probability of decoding failure through an erasure channel and developed a threshold-based transmission policy that decides between the transmission of a new packet and the retransmission of the previously transmitted one.

The most related works to our paper are \cite{ngz,gz}. The work \cite{gz} considered a similar HARQ-based status update system to ours, yet with the following differences. The authors in \cite{gz} considered a single generate-at-will source, while we consider both random arrival and generate-at-will sources as a multi-source system. Considering the random arrival sources makes the system more complicated, as the transmitter does not know the availability of the fresh packets at the subsequent slots. We study the problem of minimizing the average number of transmissions subject to the average AoI constraint, while they studied the average AoI minimization problem subject to the average number of transmissions constraint.
Similarly as in \cite{gz}, we use the CMDP approach along with the Lagrangian relaxation to solve the problem in the known environment; however, we also propose the low-complexity Lyapunov-based dynamic transmission policy.
Furthermore, for the unknown environment, they proposed a learning-based transmission policy by the Lagrangian relaxation which involves running the learning procedure for several Lagrangian multipliers. Differently, our learning-based transmission policy utilizes the Lyapunov optimization theory, and thus, the learning procedure needs to be run only once. In \cite{ngz}, which is an extension of \cite{gz}, the authors considered an HARQ-based status update system that contains one generate-at-will source and several users (destinations), in which at most one user is served at each slot. They constructed a CMDP problem for minimizing the weighted average AoI subject to the average number of transmissions constraint. They solved the CMDP problem with the Lagrangian relaxation for the known environment; for the unknown environment, they proposed different learning-based transmission policies by the Lagrangian relaxation.

\subsection{Organization}
The rest of this paper is organized as follows. The system model and problem formulation are presented in Section \ref{sys}. The CMDP formulation and its solution are presented in Section \ref{stationary}. In Section \ref{dpp}, we 
present the LC-DT and the learning-based transmission policies. Numerical results are presented in Section \ref{nrds}. Finally, concluding remarks are made in Section \ref{clcn}. 

\section{System Model and Problem Formulation}\label{sys}

\subsection{System Model}
We consider a multi-source status update system that consists of $K$ sources, one transmitter, and one receiver, as depicted in Fig.~\ref{s1}. The receiver is interested in timely information about different random processes monitored by the $K$ sources. The transmitter sends status update packets\footnote{Each status update packet contains a timestamp representing the time when the sample was generated and the measured value of the monitored process.} to the receiver through an error-prone wireless channel with the aid of an HARQ protocol. The system operates in discrete time  with unit time slots $t\in\{1,2,\ldots\}$.

\begin{figure}[t]
    \centering
    \includegraphics[width = 15cm]{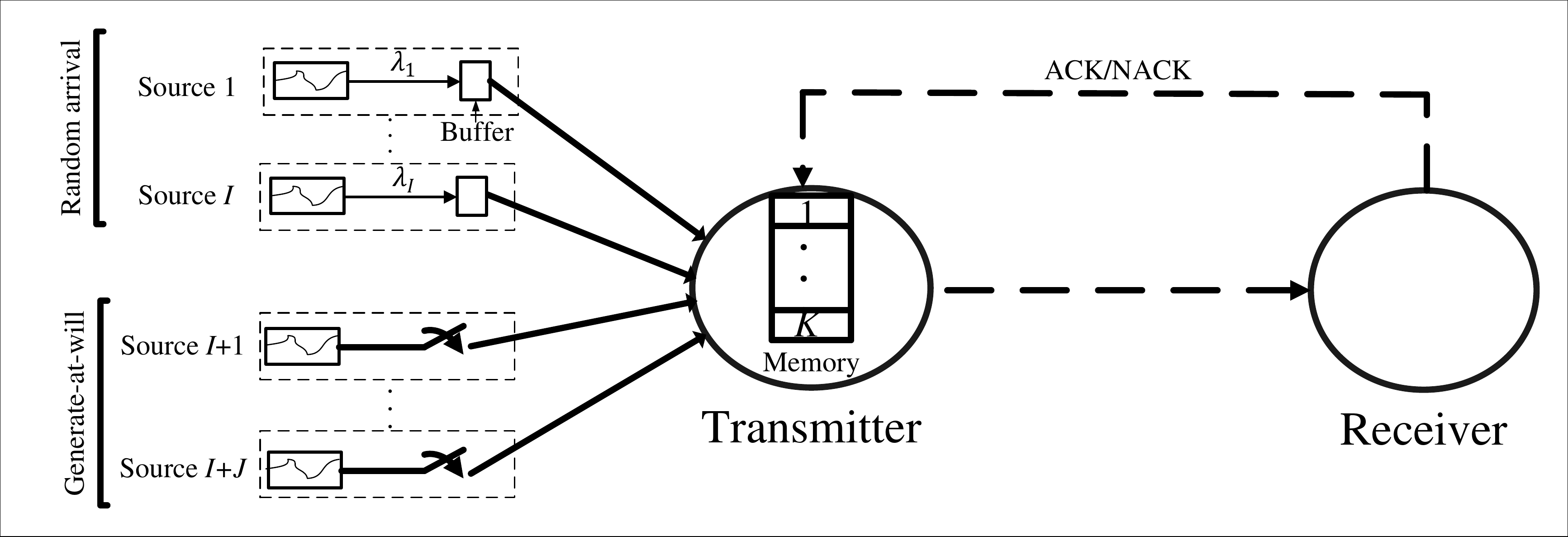}\vspace{-2mm}
    \caption{\small The considered HARQ-based multi-source status update system with two groups of sources: 1) random arrival sources with buffers which receive fresh packets with probability $\lambda_k$, 2) generate-at-will sources. The previously transmitted but not successfully decoded packets of each source are stored in the transmitter's memory. After each transmission attempt, the receiver sends a feedback signal as ACK (successful decoding) or NACK (unsuccessful decoding) to the transmitter.}
    \label{s1}\vspace{-10mm}
\end{figure}

The $K$ sources are divided into two classes based on their sampling processes: 1) a set $ \mathcal{I}$ of $I$ \textit{{random arrival}} sources whose sampling processes are uncontrollable and 2) a set $ \mathcal{J}$ of $J$ {\textit{generate-at-will}} sources, where the transmitter can sample the process at any time.
Each source $k\in\mathcal{I}$ generates status update packets randomly and independently at the beginning of slots
according to a Bernoulli random process with parameter $\lambda_k$. 
We denote the set of all sources by ${\mathcal{K} = \mathcal{I}\cup\mathcal{J} = \{1,\dots,K\}}$, where $K = I+J$.

Each random arrival source has a \textit{buffer} of size one to store the last arrived packet. As long as a new packet does not arrive, the buffer keeps the last arrived packet.
The transmitter has a {\textit{{memory}}} of size $K$ packets to store the previously transmitted but not successfully decoded packets of each source. Note that, after a number of unsuccessful transmission attempts of a packet from a source, the transmitter may decide to transmit a packet from the other sources. In this case, the transmitter retains the previously transmitted but not successfully decoded packet of each source in the transmitter's memory for possible future retransmissions, since this packet is more likely to be decoded than a new packet from that source due to the HARQ protocol. 
We term a packet in the transmitter's memory an \textit{under-process packet}. Thus, the maximum number of packets stored in the system is $I+K$ packets, i.e., $I$ packets at the buffers of random arrival sources and $K$ packets at the transmitter's memory. 

We assume that the transmitter\footnote{
A mathematically equivalent system is the one where each source is equipped with an own transmitter while at most one source is allowed to transmit at each slot.}
can transmit at most one packet per slot. 
At each slot, the transmitter decides whether to send a packet or stay idle.  
The possible transmission options for a random arrival source $k\in\mathcal{I}$ are either transmitting the packet from its buffer or retransmitting the under-process packet from the transmitter's memory. The possible transmission options for a generate-at-will source $k\in\mathcal{J}$ are either generating and  transmitting a new sample or retransmitting the under-process packet from the transmitter's memory. We refer to the packets in the buffers of the random arrival sources and to the newly generated packets of the generate-at-will sources as \textit{fresh packets}. If the transmitter decides to transmit a fresh packet 
from a given source
, this packet replaces the source's under-process packet in the transmitter's memory.
\subsubsection{Transmission Model}
At each slot $t$, the transmitter takes one of the following actions: 1) transmit a fresh packet from a source, 2) retransmit an under-process packet of a source, or 3) stay idle. 
Let $u_{t,k}\in\{0,1\}$ denote the decision variable about transmitting a fresh packet from source $k$ at slot $t$, where $u_{t,k}=1$ indicates that the transmitter sends the fresh packet, and $u_{t,k}=0$ otherwise. Let $r_{t,k}\in\{0,1\}$ denote the decision variable about retransmitting the under-process packet of source $k$ at slot $t$, where $r_{t,k}=1$ indicates that the transmitter sends the under-process packet, and $r_{t,k}=0$ otherwise. Since the transmitter can transmit at most one packet per slot, we have $\sum_{k\in \mathcal{K}}u_{t,k} + r_{t,k}\leq1$.

\textit{HARQ protocol}: In the considered HARQ protocol, every packet transmission attempt is followed by an instantaneous error-free ACK/NACK feedback signal from the receiver. 
Let $d_{t}\in\{0,1\}$ denote the packet reception status at slot $t$, where $d_{t}=1$ indicates that the transmitted packet was decoded successfully (ACK), and $d_{t}=0$ indicates that either the transmitted packet was not decoded successfully (NACK) or the transmitter remained idle. In the HARQ protocol, the receiver uses all previously received versions of a packet to decode it. Therefore, the probability of successfully decoding a packet is an increasing function of the number of attempted transmissions of the packet. Let ${x_{t,k}}$ denote the number of attempted transmissions of a packet of source $k$ up to slot $t$. The evolution of $x_{t,k}$ is given as
\begin{equation}\label{trr}
x_{t+1,k} = 
     \left\{
        \begin{array}{ll}
        1 &\quad  u_{t,k}=1 \\
        x_{t,k} &\quad  u_{t,k}+r_{t,k}=0 \\
        x_{t,k}+1 &\quad  r_{t,k}=1. \\
        \end{array}
     \right.
\end{equation}

To account for the fact that most practical HARQ protocols allow only a finite number of retransmissions, we limit the number of transmission attempts of a packet to $x^{\mathrm{max}}$, i.e., ${x_{t,k}\leq x^{\mathrm{max}}}$.
The function representing the probability of successful decoding after $x_{t,k}$ transmissions is denoted by $f(x_{t,k})$. 
In practice, $f(\cdot)$ is a complicated function of several parameters such as the channel conditions, the channel coding methods, and the combining technique utilized in the HARQ protocol \cite{frr,harqc}.
\subsubsection{Age of Information}
The AoI is defined as the time elapsed since
the generation of the most recently received status update packet at a destination.
Let $\delta_{t,k}$ denote the AoI of source $k$ at the receiver at slot $t$; we refer to this simply as the AoI of source $k$ hereinafter. We use the common assumption (see, e.g., \cite{gz,ngz,bz,mmpower,9241401}) that all AoI values in the system are upper bounded by $\delta^{\mathrm{max}}$. Besides making the analysis tractable, this supports the fact that an AoI value exceeding a high enough upper bound carries the same timeliness information as the upper bound for the destination's decision making (e.g., control actions for a drone).
To characterize the AoI of each source, we need to define the age of a fresh packet at a source and the age of an under-process packet in the transmitter's memory. These are defined in the following.

\textit{Age of the fresh packets}:
Let $\delta^{\mathrm{f}}_{t,k}$ denote the age of the fresh packet of source $k$ at slot $t$. For a random arrival source, if a packet arrives at the buffer at the beginning of slot $t$, the age of the fresh packet becomes zero, otherwise it is incremented by one. 
Let $b_{t,k}\in\{0,1\}$ denote the packet arrival status of source $k\in\mathcal{I}$ at slot $t$, where $b_{t,k} = 1$ indicates a packet arrives at the buffer, and $b_{t,k} = 0$ otherwise. Note that $\mathrm{Pr}(b_{t,k} = 1) = \lambda_k$. For the generate-at-will sources, the transmitter can generate a fresh packet at any time so that the age of the fresh packet is always zero. Thus, the evolution of {$\delta_{t,k}^{\mathrm{f}}$ with the initial value $\delta_{0,k}^{\mathrm{f}} = 0$} is given as
\begin{equation}\label{ageb}
 \delta_{t,k}^{\mathrm{f}} = 
\left\{
    \begin{array}{ll}
       0 &\quad b_{t,k} = 1,~ k\in \mathcal{I} \\
       \min\{\delta^{\mathrm{f}}_{t-1,k}+1,{\delta}^{\mathrm{max}}\} & \quad b_{t,k} = 0,~ k\in \mathcal{I}\\
       0 &\quad k\in \mathcal{J},\\
     \end{array}
    \right.
\end{equation}

\textit{Age of the under-process packets}:
Let $\delta^{\mathrm{p}}_{t,k}$ denote the age of the under-process packet of source $k$ at slot $t$. If the transmitter sends a fresh packet of source $k$ at slot $t$, the age of the under-process packet of the source  at the next slot drops to $\min\{\delta^{\mathrm{f}}_{t,k}+1,{\delta}^{\mathrm{max}}\}$. In other cases (i.e., retransmission or staying idle), the age of the under-process packet is incremented by one. 
The evolution of $\delta^{\mathrm{p}}_{t,k}$ { with the initial value $\delta_{0,k}^{\mathrm{p}} = 0$} is given by
\begin{equation}\label{ager}
\delta_{t+1,k}^{\mathrm{p}} = 
\left\{
\begin{array}{ll}
       \min\{\delta^{\mathrm{f}}_{t,k}+1,{\delta}^{\mathrm{max}}\} &\quad  u_{t,k}=1 \\
       
       \min\{\delta_{t,k}^\mathrm{p}+1,{\delta}^{\mathrm{max}}\} & \quad \mathrm{otherwise}.\\
     \end{array}
     \right.
\end{equation}

\textit{AoI at the receiver}: Having defined $\delta^{\mathrm{f}}_{t,k}$ and $\delta^{\mathrm{p}}_{t,k}$, we now characterize the evolution of the AoI at the receiver. If the transmitter sends a fresh packet of source $k$ at slot $t$ (i.e., $u_{t,k} = 1$) and the packet is decoded successfully at the receiver (i.e., $d_{t} = 1$), the AoI of the source  at the next slot drops to $\min\{\delta^{\mathrm{f}}_{t,k}+1,{\delta}^{\mathrm{max}}\}$, otherwise (i.e., $d_{t} = 0$), the AoI increases by one. 
If the transmitter retransmits the under-process packet of source $k$ (i.e., $r_{t,k} = 1$) and it is decoded successfully at the receiver, the AoI of the source at the next slot drops to $\min\{\delta^{\mathrm{p}}_{t,k}+1,{\delta}^{\mathrm{max}}\}$, otherwise (i.e., $d_{t} = 0$), the AoI increases by one. 
If, at slot $t$, the transmitter does not transmit any packet of source $k$ (i.e., $u_{t,k}+r_{t,k} = 0$), the AoI of the source at the next slot increases by one. 
The evolution of $\delta_{t,k}$ { with the initial value $\delta_{0,k}^{\mathrm{}} = 0$} is given as
\begin{equation}\label{age1}
\delta_{t+1,k} = 
\left\{
\begin{array}{ll}
       \min\{\delta^{\mathrm{f}}_{t,k} + 1,{\delta}^{\mathrm{max}}\} & \quad u_{t,k}d_{t}=1\\
       \min\{\delta^{\mathrm{p}}_{t,k} + 1,{\delta}^{\mathrm{max}}\} & \quad r_{t,k}d_{t}=1\\
       \min\{\delta_{t,k}+1,{\delta}^{\mathrm{max}}\} & \quad u_{t,k}(1-d_{t})=1\\
       \min\{\delta_{t,k}+1,{\delta}^{\mathrm{max}}\} & \quad r_{t,k}(1-d_{t})=1\\
      \min\{\delta_{t,k}+1,{\delta}^{\mathrm{max}}\} &\quad u_{t,k}+r_{t,k}=0.\\       
\end{array}
\right.
\end{equation}
Note that the conditions in \eqref{age1} are mutually exclusive and collectively exhaustive.

\subsection{Problem Formulation}
Our main goal is to minimize the average number of transmissions subject to the average AoI constraint by finding a transmission policy that determines the transmission decision variables at each slot $t$, $\{u_{t,k},r_{t,k}\}_{k\in\mathcal{K}}$.
The transmission decision for slot $t$ is based on the age of the fresh packets, $\delta^{\mathrm{f}}_{t,k}$, the age of the under-process packets, $\delta^{\mathrm{p}}_{t,k}$, the AoI of each source, $\delta_{t,k}$, and the number of previous transmission attempts of each under-process packet, $x_{t,k}$.

Let $\tau_t\in\{0,1\}$ denote the transmission status at slot $t$, where $\tau_t = 1$ indicates that the transmitter sends a packet, and $\tau_t = 0$ otherwise. Thus, we have
\begin{equation}
 \tau_t=
 \left\{
\begin{array}{ll}
     1&\quad  \textstyle\sum_{k\in \mathcal{K}}u_{t,k}+r_{t,k}=1 \\
     0&\quad \mathrm{otherwise}.
    \end{array}
\right.
\end{equation}
Let $\bar{\tau}$ denote the expected long-term time average number of transmissions, defined as
\begin{align}\label{C1}
\begin{array}{ll}
    \bar{\tau} = \underset{T\rightarrow \infty}{\limisup} \frac{1}{T} \sum_{t=1}^{T}\mathbb{E} \{{ \tau_t}\},
    \end{array}
\end{align}
where $\mathbb{E}\{\cdot\}$ is the expectation with respect to the randomness of the system (i.e., packet arrival processes of the random arrival sources and randomness in the communication channel) and the decision variables $\{u_{t,k},r_{t,k}\}_{k\in\mathcal{K}}$. Finally, let $\bar{\delta}_{}$ denote the expected long-term time average of AoI, given as
\begin{equation}\label{O1}
\begin{array}{ll}
 \bar{\delta} = \underset{T\rightarrow \infty}{\limisup}\frac{1}{T}\sum_{t=1}^{T}\mathbb{E} \{\hat\delta_{t} \},
 \end{array}
\end{equation}
where $\hat\delta_{t}$ is average AoI over all sources at slot $t$, given as
\begin{equation}\label{avokaoi}
    \begin{array}{ll}
         \hat\delta_{t} = {\frac{1}{K}\sum_{k=1}^{K} \delta_{t,k}}.
    \end{array}
\end{equation}

Using \eqref{C1} and \eqref{O1}, the main problem of this paper is formulated as the following stochastic optimization problem:
\begin{spacing}{1.3}\vspace{-5mm}
\begin{subequations}\label{p1}
\begin{alignat}{2}
\amin\quad & \bar{\tau}\label{p1:m}\\
\mbox{subject to} \quad & \bar{\delta}\leq{\Delta}^{\mathrm{max}}\label{p1:1}\\
&x_{t+1,k}\leq x^{\mathrm{max}},~ k\in \mathcal{K},~t\in \mathbb{N}\label{p1:2}\\
&\textstyle\sum_{k\in \mathcal{K}}u_{t,k} + r_{t,k}\leq1, ~t\in \mathbb{N}\label{p1:3}\\
&u_{t,k},r_{t,k}\in \{0,1\},~ k\in \mathcal{K},~t\in \mathbb{N},\label{p1:4}
\end{alignat}
\end{subequations}
\end{spacing}
\noindent with variables $\{u_{t,k},r_{t,k}\}_{k\in\mathcal{K}}$ for
all $t\in\mathbb{N}$,
where  $\Delta^{\mathrm{max}}$ is the maximum allowed average AoI. The constraints of problem \eqref{p1} are as follows. Inequality \eqref{p1:1} represents the average AoI constraint. Inequality \eqref{p1:2} ensures that the number of transmission attempts for each packet cannot exceed $x^{\mathrm{max}}$. Inequality \eqref{p1:3} ensures that the transmitter can transmit at most one packet per slot. Expression \eqref{p1:4} indicates the binary nature of the decision variables.



\section{Deterministic Transmission policy}\label{stationary}
In this section, we propose a (near-optimal) solution to main problem \eqref{p1} for the known environment, i.e., the packet arrival probability of each random arrival source, $\lambda_k$, and the probability of successful decoding function, $f(\cdot)$, are known. We cast problem \eqref{p1} as a constrained Markov decision process (CMDP) problem. Then, we use the Lagrangian relaxation 
{to find a near-optimal deterministic transmission policy.}
\subsection{CMDP Formulation}\label{cmdpfm}
The CMDP is defined by a tuple of five elements $(\mathcal{S},\mathcal{A}_s,\mathcal{P}, c,d)$: state space, action space, state transition probabilities, and two cost functions, which are defined in the following.

{\textit{State}}: Let ${s_{t,k}=\{ \delta^{\mathrm{f}}_{t,k},\delta^{\mathrm{p}}_{t,k},\delta_{t,k},x_{t,k}\}}$ denote the state of source $k$ at slot $t$. The system state at slot $t$ is defined as ${s_{t}=\{s_{t,k}\}_{k\in \mathcal{K}}}\in \mathcal{S}$, where $\mathcal{S}$ is the state space. The initial state is denoted with $s_0=\{s_{0,k}\}_{k\in \mathcal{K}}$, where $s_{0,k}=\{0,0,0,0\}$ for all $k\in \mathcal{K}$.

{\textit{Action}}:
Let $a_t=\{a_{t,k}\}_{k\in\mathcal{K}}\in\mathcal{A}_{s_t}$ denote the action of the transmitter at slot $t$, where ${a_{t,k}=\{u_{t,k},r_{t,k}\}}$ represents the action for source $k$, and $\mathcal{A}_{s_t}$ is a space of feasible actions in state $s_t$, defined as $\mathcal{A}_{s_t}= \big\{u_{t,k},r_{t,k}\in\{0,1\} \mid k\in \mathcal{K},~\sum_{k\in \mathcal{K}}u_{t,k}+r_{t,k}\leq 1,~r_{t,k}(x_{t,k}+1)\leq x^{\mathrm{max}}\big\}$.

{\textit{Cost functions}}: The CMDP has two cost functions: 1) transmission cost, defined as ${c(a_t) = \tau_t}$, i.e., $c(a_t)=1$ if the transmitter makes a transmission attempt at slot $t$, otherwise $c(a_t)=0$, and 2) AoI cost, defined as $d(s_t) = \hat{\delta}_{t}$, i.e., the average AoI over sources at slot $t$.

{\textit{State transition probabilities}}: Let ${\mathcal{P}(s'\mid s,a)=\mathrm{Pr}(s'=s_{t+1}\mid s=s_t,a=a_t)}$ denote the state transition probabilities, defined as the probability of moving from current state ${s=s_t}$ to a next state ${s'=s_{t+1}}$ under action $a=a_t$. Given an action, the one-slot evolution of the AoI values (at the source, memory, and destination) and of the number of transmissions of the under-process packets is independent among the sources.
Therefore, the state transition probability factorizes as ${\mathcal{P}(s'\mid s,a)=\Pi_{k\in\mathcal{K}}\mathrm{Pr}(s_{t+1,k}\mid s_{t,k},a_{t,k})}$.
Let us denote ${\tilde\delta_{t,k}^{\mathrm{f}} \triangleq \min\{\delta_{t,k}^{\mathrm{f}}+1,\delta^{\mathrm{max}}\}}$, ${\tilde\delta_{t,k}^{\mathrm{p}} \triangleq \min\{\delta_{t,k}^{\mathrm{p}}+1,\delta^{\mathrm{max}}\}}$, ${\tilde\delta_{t,k}^{\mathrm{}} \triangleq \min\{\delta_{t,k}^{\mathrm{}}+1,\delta^{\mathrm{max}}\}}$, ${\bar f(\cdot) \triangleq 1-f(\cdot)}$, and ${\bar \lambda_k \triangleq 1-\lambda_k}$. Given the state $s_{t,k}=\{ \delta^{\mathrm{f}}_{t,k},\delta^{\mathrm{p}}_{t,k},\delta_{t,k},x_{t,k}\}$, the state transition probabilities for a random arrival source $k\in\mathcal{I}$ under different actions can be expressed as
\begin{spacing}{1.3}\vspace{-5mm}
\begin{subequations}\label{prt}
    \begin{alignat}{2}
    \label{prt1}
    &\mathrm{Pr}\big(\{0,\tilde\delta_{t,k}^{\mathrm{f}},\tilde\delta_{t,k}^{\mathrm{f}},1\}\mid s_{t,k},a_{t,k}=\{1,0\}\big)=f(1)\lambda_k\\
    &\mathrm{Pr}\big(\{\tilde\delta_{t,k}^{\mathrm{f}},\tilde\delta_{t,k}^{\mathrm{f}},\tilde\delta_{t,k}^{\mathrm{f}},1\}\mid s_{t,k},a_{t,k}=\{1,0\}\big)=f(1)\bar \lambda_k\\
    &\mathrm{Pr}\big(\{0,\tilde\delta_{t,k}^{\mathrm{f}},\tilde\delta_{t,k}^{\mathrm{}},1\}\mid s_{t,k},a_{t,k}=\{1,0\}\big)=\bar f(1)\lambda_k\\
    &\mathrm{Pr}\big(\{\tilde\delta_{t,k}^{\mathrm{f}},\tilde\delta_{t,k}^{\mathrm{f}},\tilde\delta_{t,k}^{\mathrm{}},1\}\mid s_{t,k},a_{t,k}=\{1,0\}\big)=\bar f(1)\bar \lambda_k\\ 
    &\mathrm{Pr}\big(\{0,\tilde\delta_{t,k}^{\mathrm{p}},\tilde\delta_{t,k}^{\mathrm{p}},x_{t,k}+1\}\mid s_{t,k},a_{t,k}=\{0,1\}\big)=f(x_{t,k}+1)\lambda_k\\
    &\mathrm{Pr}\big(\{\tilde\delta_{t,k}^{\mathrm{f}},\tilde\delta_{t,k}^{\mathrm{p}},\tilde\delta_{t,k}^{\mathrm{p}},x_{t,k}+1\}\mid s_{t,k},a_{t,k}=\{0,1\}\big)=f(x_{t,k}+1)\bar \lambda_k\\
    &\mathrm{Pr}\big(\{0,\tilde\delta_{t,k}^{\mathrm{p}},\tilde\delta_{t,k}^{\mathrm{}},x_{t,k}+1\}\mid s_{t,k},a_{t,k}=\{0,1\}\big)=\bar f(x_{t,k}+1)\lambda_k\\\label{prt8}
    &\mathrm{Pr}\big(\{\tilde\delta_{t,k}^{\mathrm{f}},\tilde\delta_{t,k}^{\mathrm{p}},\tilde\delta_{t,k}^{\mathrm{}},x_{t,k}+1\}\mid s_{t,k},a_{t,k}=\{0,1\}\big)=\bar f(x_{t,k}+1)\bar \lambda_k\\\label{prt9}
    &\mathrm{Pr}\big(\{0,\tilde\delta_{t,k}^{\mathrm{p}},\tilde\delta_{t,k}^{\mathrm{}},x_{t,k}\}\mid s_{t,k},a_{t,k}=\{0,0\}\big)=\lambda_k\\\label{prt10}
    &\mathrm{Pr}\big(\{\tilde\delta_{t,k}^{\mathrm{f}},\tilde\delta_{t,k}^{\mathrm{p}},\tilde\delta_{t,k}^{\mathrm{}},x_{t,k}\}\mid s_{t,k},a_{t,k}=\{0,0\}\big)=\bar \lambda_k,
    \end{alignat}
\end{subequations}
\end{spacing}
\noindent whereas the other cases are zero. The state transition probabilities for the generate-at-will source $k\in\mathcal{J}$ are obtained by substituting $\lambda_k=1$ in \eqref{prt}. 


Let $\pi$ denote a policy that determines the action taken at each state. A stationary randomized policy is mapping from each state to a distribution over actions, ${\pi(a\mid s): \mathcal{S}\times \mathcal{A} \rightarrow [0,1],~ \sum_{a \in \mathcal{A}_s} \pi(a\mid s) =1}$. A (stationary) deterministic policy chooses an action at a given state with probability one, which is a special case of the stationary randomized policy. With a slight abuse of notation, we denote the action taken in state $s$ by a deterministic policy $\pi$ with $\pi(s)$. 
Let ${\bar{\tau}^{\pi}=\underset{T\rightarrow \infty}{\limisup} \frac{1}{T} \sum_{t=1}^{T}\mathbb{E} \{{c(a_t)}\mid s_0\}}$ denote the average number of transmissions (see \eqref{C1}), obtained under policy $\pi$ starting from the initial state $s_0$. Let ${\bar{\delta}^{\pi}=\underset{T\rightarrow \infty}{\limisup}\frac{1}{T}\sum_{t=1}^{T}\mathbb{E} \{d(s_t)\mid s_0\}}$ denote average AoI (see \eqref{O1}), obtained under policy $\pi$ starting from the initial state $s_0$. 
Having constructed the CMDP, problem \eqref{p1} is equivalently cast as the CMDP problem
\begin{equation}\label{pcmdp01}
\begin{array}{ll}
\displaystyle\amin_{\pi} \quad & \bar{\tau}^{\pi}\\
\mbox{subject to} \quad & \bar{\delta}^{\pi}\leq {\Delta}^{\mathrm{max}}
.
\end{array}
\end{equation}
An optimal policy that solves CMDP problem \eqref{pcmdp01} is denoted with $\pi^*$, and the optimal value of the problem is denoted with $\bar\tau^*$.

{Similarly to \cite{shr,bz,kr}, to solve problem \eqref{pcmdp01}, we need to make extra assumptions about the CMDP structure. Specifically, we assume that given the initial state ($s_0$), all policies will induce a Markov chain with only one recurrent class and a (possibly empty) set of transient states. This assumption makes problem \eqref{pcmdp01} well-posed so that we can use the tools associated with the unichain MDPs, as described in the next section. 
}

\subsection{Solution of the CMDP Problem}
{To solve CMDP problem \eqref{pcmdp01}, we apply the Lagrangian relaxation method to transform the CMDP problem to an (unconstrained) MDP problem, parametrized by a Lagrange dual variable \cite[Sec. 3.3]{elman}. {In comparison to the CMDP problem, the MDP problem has only one cost function that is defined as ${L(s,a,\beta)=c(a_t)+\beta d(s_t)}$, whereas the other elements, i.e., the state space, action space, and state transition probabilities, are the same.}
Let ${\bar{L}(\pi, \beta)= \limsup_{T \rightarrow \infty} \frac{1}{T}\sum_{t=1}^{T}\mathbb{E} \{c(a_t) + \beta{\big(d(s_t)-\Delta^{\mathrm{max}}\big)}\}}$ denote the Lagrangian corresponding to CMDP problem \eqref{pcmdp01}, where $\beta$ is the Lagrangian multiplier. {Following the standard Lagrangian relaxation procedure, we restrict to the set of deterministic policies and construct the following MDP problem associated with the CMDP problem \eqref{pcmdp01}}}
\begin{equation}\label{plagl}
\begin{aligned}
\amin_{\pi{\in\Pi_{\mathrm{D}}}} \quad &\bar{L}(\pi, \beta)
,\\
\end{aligned}
\end{equation}
{where $\Pi_{\mathrm{D}}$ is the set of all deterministic policies. 
Let $\pi_{\beta}^*$ denote an optimal policy that solves problem \eqref{plagl} for a given $\beta$, which is called a \textit{$\beta$-optimal} policy.}

The following remark expresses the relation between the optimal values of CMDP problem \eqref{pcmdp01} and the MDP problem \eqref{plagl}.

{\remark{The cost function in the objective of CMDP problem \eqref{pcmdp01} is bounded below, i.e., $c(a_t)\geq0$ for all $t\in\mathbb{N}$. Moreover, the state space, $\mathcal{S}$, is finite. Therefore, the two conditions in \cite[Corollary 12.2]{elman} are satisfied in our CMDP formulation, and we have
\begin{equation}\label{rm1eq}
    \bar\tau^{*} = \sup_{\beta\geq 0}\min_{\pi\in\Pi_{\mathrm{D}}}\bar{L}(\pi, \beta).
\end{equation}
}\label{rm1}}

According to Remark \ref{rm1}, the optimal value of CMDP problem \eqref{pcmdp01}, $\bar\tau^{*}$, is obtained via the solution of the right hand side of \eqref{rm1eq}, which means finding the optimal Lagrangian multiplier $\beta^*$ and its corresponding $\beta^*$-optimal policy, ${\pi^*_{\beta^*}}$. 
If policy ${\pi^*_{\beta^*}}$ satisfies the constraint of CMDP problem \eqref{pcmdp01} with equality, i.e., ${\bar{\delta}^{\pi^*_{\beta^*}}= {\Delta}^{\mathrm{max}}}$, then ${\pi^*_{\beta^*}}$ is an optimal policy for the CMDP problem, i.e., $\pi^* = {\pi^*_{\beta^*}}$. However, due to the discrete nature of the action space, in general, there is no guarantee that ${\pi^*_{\beta^*}}$ satisfies the constraint with equality. 
To elaborate this further, the following remark presents the structure of an optimal policy $\pi^*$.
{\remark{An optimal policy for CMDP problem \eqref{pcmdp01}, $\pi^*$, is a randomized mixture of two deterministic $\tilde\beta$-optimal policies, from which one policy satisfies the constraint and the other one violates it. The two policies are mixed with a randomization factor such that the obtained optimal policy satisfies ${\bar{\delta}^{\pi^*}= {\Delta}^{\mathrm{max}}}$ \cite{beutler,bz,gz}.}\label{rmk2}
}

According to Remark \ref{rmk2}, two deterministic $\tilde\beta$-optimal policies and the optimal randomization factor to mix between these policies should be found to obtain an optimal policy, $\pi^*$. However, finding these becomes readily computationally intractable even for the moderate number of states, especially because oftentimes, the optimal randomization factor can be found only numerically \cite[Section 3.2]{nt}. Therefore, in order to solve CMDP problem \eqref{pcmdp01}, we propose a practical deterministic policy, which is numerically shown to provide near-optimal performance in Section \ref{nrds}. More specifically, we develop an iterative algorithm based on bisection and the relative value iteration algorithm (RVIA), as summarized in Algorithm \ref{Acmdp}. In brief, at each iteration, we find a $\beta$-optimal policy for a given $\beta$ via the RVIA and subsequently update $\beta$ according to the bisection rule. The iterative procedure continues until the best $\beta$-optimal policy among the feasible $\beta$-optimal policies is found. In the next two subsections, we delve into details of this procedure.   

\subsubsection{Algorithm to Find a $\beta$-optimal Policy} {To obtain an optimal policy $\pi_{\beta}^*$ for a given $\beta$, we solve the MDP problem \eqref{plagl} via RVIA. 
By \cite[Theorem 8.4.3]{puterman1994}, there exists a relative value function $h(s),~s\in\mathcal{S}$, that satisfies}
\begin{equation}\label{eq3}
    \bar{L}^*( \beta) + h(s) = \min_{a\in \mathcal{A}_s}\big[L(s,a, \beta) + \sum_{s'\in \mathcal{S}}\mathrm{Pr}(s' \mid s,a)h(s')  \big],~ \text{for all }s\in \mathcal{S},
\end{equation}
where $\bar{L}^*( \beta)$ is the optimal value of the MDP problem \eqref{plagl} for a given $\beta$, defined as ${\bar{L}^*( \beta)=\min_{\pi\in\Pi_{\mathrm{D}}}\bar{L}(\pi, \beta)}$.
Subsequently, the $\beta$-optimal policy, $\pi_{\beta}^*$, is obtained as \cite[Theorem 8.4.4]{puterman1994}
\begin{equation}\label{qopt}
    \pi_{\beta}^*(s) = \argmin_{a\in \mathcal{A}_s} \big[L(s,a, \beta) + \sum_{s'\in \mathcal{S}}\mathrm{Pr}(s' \mid s,a)h(s') \big],~ \text{for all }{s} \in \mathcal{S}.
\end{equation}

To obtain the $\beta$-optimal policy, we use the RVIA, in which the relative value function for all states $s\in\mathcal{S}$ at each iteration $i\in\{0,1,\dots\}$ is updated as $h^{i}(s) = v^{i}(s)-v^{i}(s^{\mathrm{ref}})$.
Where $s^{\mathrm{ref}}\in\mathcal{S}$ is an arbitrary reference state which remains unchanged throughout the iterations. The term $v^{i}(s)$, called value function, is obtained at each iteration as
\begin{equation}\label{rvif}
    \begin{array}{ll}
         v^{i}(s)=\min_{a\in \mathcal{A}_s}\big[L(s,a, \beta) +\sum_{s'\in \mathcal{S}}\mathrm{Pr}(s' \mid s,a)h^{i-1}(s')\big].
    \end{array}
\end{equation}
For any state $s\in \mathcal{S}$ and initialization $v^0(s)$, the sequences $\{h^i(s)\}_{i=1,2,\dots}$ and $\{v^i(s)\}_{i=1,2,\dots}$ converge, i.e., $\lim_{i\rightarrow\infty} h^i(s) = h(s)$ and $\lim_{i\rightarrow\infty} v^i(s) = v(s)$.
The RVI algorithm to find a $\beta$-optimal policy is presented in  Steps $3$-$12$ of Algorithm \ref{Acmdp}. After the convergence of RVIA, i.e., convergence of the relative value function, $h(\cdot)$, and the value function, $v(\cdot)$ (see Steps $3$-$9$ in Algorithm \ref{Acmdp}), we obtain the $\beta$-optimal policy, $\pi^*_\beta$, according to \eqref{qopt} (see Steps $10$-$12$ in Algorithm \ref{Acmdp}). It is worth noting that the optimal value of the MDP problem \eqref{plagl} for a given $\beta$ is given by ${\bar L^*( \beta)=v(s^{\mathrm{ref}})}$.

\subsubsection{Algorithm to Find the Optimal Lagrangian Multiplier} According to \cite[Lemma 3.1]{beutler}, for a given ${\beta\text{-optimal}}$ policy (${\pi^*_{\beta}}$), the objective function of the CMDP problem, $\bar{\tau}^{\pi^*_{\beta}}$, and the objective function of the MDP problem, $\bar{L}^{*}(\beta)$, are increasing in $\beta$, while the constraint of the CMDP problem, $\bar{\delta}^{\pi^*_{\beta}}$, is decreasing in $\beta$. Therefore, we are interested in the smallest Lagrangian multiplier that satisfies the constraint in CMDP problem \eqref{pcmdp01}, defined as
\begin{equation}\label{betastar}
    \begin{aligned}
    \tilde\beta \triangleq\inf~\{\beta\geq 0 \mid {\bar{\delta}^{\pi^*_{\beta}}}\leq \Delta^{\mathrm{max}}\}.
    \end{aligned}
\end{equation}

To search for $\tilde\beta$, we use the bisection algorithm which takes advantage of the monotonicity of $\bar{\delta}^{\pi^*_{\beta}}$ with respect to $\beta$, as presented in Algorithm \ref{Acmdp} (see Steps $1$-$18$). We initialize the bisection algorithm with $\beta_{\mathrm{u}}$ and $\beta_{\mathrm{l}}$ in such a way that ${\bar{\delta}^{\pi^*_{\beta_{\mathrm{u}}}}}\leq \Delta^{\mathrm{max}}$ and ${\bar{\delta}^{\pi^*_{\beta_{\mathrm{l}}}}}\geq \Delta^{\mathrm{max}}$, which also implies $\beta_{\mathrm{u}}\geq\beta_{\mathrm{l}}$. The algorithm termination criterion is $\beta_{\mathrm{u}}-\beta_{\mathrm{l}}<\kappa$, where $\kappa$ is a sufficiently small constant.
After termination of the bisection algorithm, we set $\tilde\beta = \beta_{\mathrm{u}}$ and obtain the best feasible $\beta$-optimal policy as $\pi^*_{\tilde\beta} = \pi^*_{\beta_{\mathrm{u}}}$. Moreover, the algorithm returns the infeasible policy associated with $\beta_{\mathrm{l}}$, which represents a lower-bound to an optimal solution of \eqref{pcmdp01}.

\begin{algorithm}[t]
    \SetKwInOut{Inputi}{Initialize}
    \SetKwInOut{output}{Output}
    \SetKwInOut{Output}{Output}
    \SetKwComment{Comment}{/*}{ }
    \KwIn{1) System parameters: $\Delta^{\mathrm{max}}$, $f(\cdot)$, and $\lambda_k$ for all $k\in \mathcal{I}$, 2) RVI parameters: $s^{\mathrm{ref}}$, $\epsilon$, and 3) Bisection parameters: $\beta_{\mathrm{u}}$, $\beta_{\mathrm{l}}$, $\kappa$}
    \Comment{Bisection algorithm}
    \While{$ \beta_{\mathrm{u}}-\beta_{\mathrm{l}}\geq\kappa$}
    {
    $\bar{\beta} = \frac{\beta_{\mathrm{u}}+\beta_{\mathrm{l}}}{2}$\\
    \Comment{RVIA for the given $\bar \beta$}
    \Inputi{1) $i=1$, 2) set $h^0(s) =1, h^1(s) =0, v^0(s) = 0$ for all $s\in\mathcal{S}$}
    \While{$\max_{s\in\mathcal{S}}|h^{i}(s)-h^{i-1}(s)|\geq\epsilon$}{
    $i = i+1$\\
    \For{$s\in\mathcal{S}$}{$v^{i}(s)=\min_{a\in \mathcal{A}_s}\big[L(s,a, \bar\beta) +\sum_{s'\in \mathcal{S}}\mathrm{Pr}(s' \mid s,a)h^{i-1}(s')\big]$\\
    $h^{i}(s) = v^{i}(s)-v^{i}(s^{\mathrm{ref}})$}
    }
    
    \Comment{An optimal policy for given $\bar \beta$}
    \For{$s\in\mathcal{S}$}{$\pi^*_{\bar\beta}(s)=\argmin_{a\in \mathcal{A}_s}\big[L(s,a, \bar\beta) +\sum_{s'\in \mathcal{S}}\mathrm{Pr}(s' \mid s,a)h^i(s')\big]$}
    \eIf{$\bar\delta^{\pi^*_{\bar\beta}}\leq \Delta^{\mathrm{max}}$}
    {
    $\beta_{\mathrm{u}} = \bar\beta$\\
    }
    {
    $\beta_{\mathrm{l}} = \bar\beta$\\
    }
    }
    \KwOut{Lagrangian multiplier: $\tilde\beta = \beta_{\mathrm{u}}$, feasible policy: $\pi^*_{\tilde\beta} = \pi^*_{\beta_{\mathrm{u}}}$, (infeasible) lower-bound policy: $\pi^*_{\beta_{\mathrm{l}}}$} 
    \caption{The deterministic transmission policy to solve CMDP problem \eqref{pcmdp01}}
    \label{Acmdp}
\end{algorithm}

\section{Lyapunov-based Transmission Scheduling Policies}\label{dpp}
In this section, we use the Lyapunov optimization theory to derive a solution for problem \eqref{p1}. 
According to the deterministic transmission policy (Algorithm \ref{Acmdp}), RVIA needs to explore all states and actions. When the number of sources increases, the number of states grows exponentially. Besides this, the RVIA is run for each bisection iteration. These make obtaining the policy computationally inefficient.
Therefore, we develop a low-complexity dynamic transmission (LC-DT) policy using the DPP method for the known environment in Section \ref{known}. The numerical results in Section \ref{nrds} show that LC-DT policy performs close to the optimal solution.

Furthermore, we develop a Lyapunov-based transmission policy for the unknown environment in Section \ref{unknown}. We transform CMDP problem \eqref{pcmdp01} into an MDP problem using the DPP method. Then, by using deep Q-learning (DQL), we solve the MDP problem and provide the learning-based transmission policy.

\subsection{Low-complexity Dynamic Transmission Policy: Known Environment}\label{known}
We use the Lyapunov drift-plus-penalty (DPP) method \cite{lyp}, where average AoI constraint \eqref{p1:1} is enforced by transforming it into a queue stability constraint. In particular, the constraint is mapped into a virtual queue so that the stability of the virtual queue implies the feasibility of the constraint.

Let $Q_t$ denote the virtual queue associated with average AoI constraint \eqref{p1:1} at slot $t$.
The virtual queue evolves as follows:
\begin{equation}\label{vq1}
\begin{array}{ll}
    Q_{t+1} = \mathrm{max}\{ Q_{t}  - {\Delta}^{\mathrm{max}}+ \hat{\delta}_{t+1}, 0 \}.
    \end{array}
\end{equation}

To ensure that average AoI constraint \eqref{p1:1} is satisfied, we use the notion of strong stability, where the virtual queue is (strongly) stable if $\underset{T\rightarrow\infty}{\lim}\frac{1}{T}\sum_{t=1}^{T}\mathbb{E}\{Q_{t}\}<\infty$ \cite[Chapter 2]{lyp}. According to the DPP method, the strong stability of queue \eqref{vq1} implies that average AoI constraint \eqref{p1:1} is satisfied \cite[Chapter 4]{lyp}.

To define the queue stability condition, we introduce the Lyapunov function and its drift. A quadratic Lyapunov function is defined as ${L(Q_{t}) = \frac{1}{2}Q_{t}^2}$ \cite[Chapter 3]{lyp}.
The Lyapunov function measures the network congestion and thus, by minimizing the expected change of the Lyapunov function from one slot to the next slot, the virtual queue can be stabilized  \cite[Chapter 4]{lyp}.

Let $o_{t} = \big\{\{\delta_{t,k}^{\mathrm{f}}, \delta_{t,k}^{\mathrm{p}},\delta_{t,k},x_{t,k}\}_{k\in\mathcal{K}},Q_{t}\big\}$ denote the network state at slot $t$.
The conditional Lyapunov drift, $\alpha(o_{t})$, is defined as the expected change in the Lyapunov function over one slot, given the network state at slot $t$, i.e., $
\alpha(o_{t}) = \mathbb{E}\big \{L(Q_{t+1}) -L(Q_{t})\mid o_t\big \}$  \cite[Chapter 4]{lyp}.

By following the DPP minimization approach \cite[Chapter 3]{lyp}, a solution for \eqref{p1} can be derived by solving the following problem at each slot $t$
\begin{spacing}{1.3}\vspace{-5mm}
\begin{subequations}\label{lyp1}
\begin{alignat}{2}
\amin
\quad & V\mathbb{E}\{\tau_{t}\mid o_t \}+\alpha(o_{t})\label{lyp1:m}\\
\mbox{subject to}\quad& x_{t+1,k}\leq x^{\mathrm{max}},~ k\in\mathcal{K}\label{lyp1:2}\\
&\textstyle\sum_{k\in \mathcal{K}}u_{t,k} + r_{t,k}\leq1 \label{lyp1:3}\\
&u_{t,k},r_{t,k}\in \{0,1\},~ k\in \mathcal{K},\label{lyp1:4}
\end{alignat}
\end{subequations}
\end{spacing}
\noindent with variables $\{u_{t,k},r_{t,k}\}_{k\in\mathcal{K}}$. The objective function of problem \eqref{lyp1} represents the DPP function, in which the positive parameter $V$ is used to adjust the tradeoff between minimizing original objective function \eqref{p1:m} and the size of the (virtual) queue backlog. A larger value of $V$ puts more emphasis on the original objective function, i.e., minimizing the average number of transmissions.

According to the standard procedure used in the DPP method, an upper bound for the drift part is derived, whereas the penalty part (i.e., the original objective function) remains unchanged \cite[Chapter 4]{lyp}. It is worth stressing that using such upper bound of the conditional Lyapunov drift in the optimization procedure does not affect the virtual queue’s stabilizing logic. We derive the upper bound using the following inequality, where for any $\tilde\gamma \geq 0$, $\bar\gamma \geq 0$, and $\hat\gamma \geq 0$, we have \cite{mmpower}
\begin{equation}\label{upq}
\begin{array}{ll}
\big(\mathrm{max}\{\tilde\gamma  - \bar\gamma+ \hat\gamma , 0\}\big)^2 \leq \tilde\gamma  ^2 +  \bar\gamma ^2 + \hat\gamma  ^2 +2\tilde\gamma (\hat\gamma  - \bar\gamma).
\end{array}
\end{equation}

By applying \eqref{upq} to \eqref{vq1}, an upper bound for $Q^2_{t+1}$ is given as
\begin{equation}\label{uneq}
\begin{array}{ll}
& Q^2_{t+1} \leq Q^2_{t} + ({\Delta}^{\mathrm{max}})^2 + \hat{\delta}^2_{t + 1} + 2Q_t \big(\hat{\delta}_{t + 1} - {\Delta}^{\mathrm{max}} \big).
\end{array}
\end{equation}

By applying \eqref{uneq} to (15), the upper bound of objective function \eqref{lyp1:m} is given as
\begin{equation}\label{eqQ}
\begin{array}{ll}
   &V\mathbb{E}\{\tau_{t}\mid o_t \} + \alpha(o_{t})\leq  V\sum_{k\in \mathcal{K}}\mathbb{E}\{u_{t,k}\mid o_t\}+\mathbb{E}\{r_{t,k}\mid o_t\}  + \frac{1}{2}\mathbb{E}\big\{  ({\Delta}^{\mathrm{max}})^2 + \hat{\delta}^2_{t + 1} \\
   &+ 2Q_{t} \big(\hat{\delta}_{t + 1} - {\Delta}^{\mathrm{max}} \big)\mid o_t \big\}= V\sum_{k\in \mathcal{K}}\mathbb{E}\{u_{t,k}\mid o_t\}+ \mathbb{E}\{r_{t,k}\mid o_t\}\\
   &\hspace{1mm}+\frac{1}{2} \big(({\Delta}^{\mathrm{max}})^2 + \mathbb{E}\{ \hat{\delta}^2_{t + 1} \mid o_t\}+2Q_{t} \big(\mathbb{E}\{\hat{\delta}_{t + 1}\mid o_t\} - {\Delta}^{\mathrm{max}} \big).
\end{array}
\end{equation}

To complete the derivation of \eqref{eqQ}, we calculate $\mathbb{E}\{\hat\delta_{t + 1}\mid o_t\}$ and $\mathbb{E}\{\hat\delta_{t + 1}^2\mid o_t\}$, which are given by the following lemmas.
\begin{lemma}\label{lm1}
The conditional expectation $\mathbb{E}\{\hat\delta_{t + 1}\mid o_t\}$ is given as
\begin{equation}\label{eaoi2}
    \begin{array}{ll}
    
    \mathbb{E}\big\{\hat\delta_{t + 1}\mid o_t\big\} \hspace{-0mm}=\frac{1}{K}\sum_{k\in\mathcal{K}}\mathbb{E}\{u_{t,k}\mid o_t\}f(1)\tilde\delta_{t,k}^{\mathrm{f}} + \mathbb{E}\{r_{t,k}\mid o_t\}f(x_{t,k}+1)\tilde\delta_{t,k}^{\mathrm{p}}\\
    \hspace{1mm}+\big[1-f(1)\mathbb{E}\{u_{t,k}\mid o_t\}-f(x_{t,k}+1)\mathbb{E}\{r_{t,k}\mid o_t\}\big]\tilde\delta_{t,k}.
    
    \end{array}
\end{equation}
\end{lemma}
\begin{proof}
The proof is presented in Appendix \ref{plm1}. 
\end{proof}
\begin{lemma}\label{lm2}
The conditional expectation $\mathbb{E}\{\hat\delta_{t + 1}^2\mid o_t\}$ is given as
\begin{equation}\label{eaoi2}
    \begin{array}{ll}
    \mathbb{E}\big\{\hat\delta_{t + 1}^2\mid o_t\big\} \hspace{-0mm}=\frac{1}{K^2}\Big[\sum_{k\in\mathcal{K}}\mathbb{E}\{u_{t,k}\mid o_t\}f(1)(\tilde\delta_{t,k}^{\mathrm{f}})^2 + \mathbb{E}\{r_{t,k}\mid o_t\}f(x_{t,k}+1)(\tilde\delta_{t,k}^{\mathrm{p}})^2\\
    \hspace{1mm}+\big[1-f(1)\mathbb{E}\{u_{t,k}\mid o_t\}-f(x_{t,k}+1)\mathbb{E}\{r_{t,k}\mid o_t\}\big](\tilde\delta_{t,k} )^2\\ \hspace{1mm} + \sum_{k\in\mathcal{K}}\sum_{k'\in\mathcal{K},k'\neq k}\mathbb{E}\{u_{t,k}\mid o_t\}f(1)\tilde\delta_{t,k}^{\mathrm{f}}\tilde\delta_{t,k'} + \mathbb{E}\{r_{t,k}\mid o_t\}f(x_{t,k}+1)\tilde\delta_{t,k}^{\mathrm{p}} \tilde\delta_{t,k'} \\ \hspace{1mm}+ \mathbb{E}\{u_{t,k'}\mid o_t\}f(1)\tilde\delta_{t,k'}^{\mathrm{f}}\tilde\delta_{t,k} + \mathbb{E}\{r_{t,k'}\mid o_t\}f(x_{t,k'}+1)\tilde\delta_{t,k'}^{\mathrm{p}} \tilde\delta_{t,k} - \mathbb{E}\{u_{t,k}\mid o_t\} f(1)\tilde\delta_{t,k} \tilde\delta_{t,k'} \\\hspace{1mm}-\mathbb{E}\{r_{t,k}\mid o_t\} f(x_{t,k}+1)\tilde\delta_{t,k} \tilde\delta_{t,k'} -\mathbb{E}\{u_{t,k'}\mid o_t\} f(1)\tilde\delta_{t,k'} \tilde\delta_{t,k} \\
    \hspace{1mm}-\mathbb{E}\{r_{t,k'}\mid o_t\} f(x_{t,k'}+1)\tilde\delta_{t,k'} \tilde\delta_{t,k} +\tilde\delta_{t,k'} \tilde\delta_{t,k}\Big].
    
    \end{array}
\end{equation}
\end{lemma}
\begin{proof}
The proof is presented in Appendix \ref{plm2}.
\end{proof}


Having derived the upper bound of the DPP in \eqref{eqQ}, our goal is to minimize \eqref{eqQ} subject to constraints \eqref{lyp1:2}--\eqref{lyp1:4} with variables $\{ u_{t,k},r_{t,k}\}_{k\in \mathcal{K}}$.
According to the standard procedure, we drop the expectations in \eqref{eqQ} \cite[Page 59]{lyp}. 
We denote the upper bound of the DPP after dropping the expectations at slot $t$ by $W_t$, which, as shown in Appendix~\ref{pfp}, can be expressed as 
\begin{equation*}\label{eqQ5n}
\begin{array}{ll}
   W_t
    &\hspace{-3mm}= V\sum_{k\in \mathcal{K}}u_{t,k}+r_{t,k} +\frac{1}{2K^2}\Big[\sum_{k\in \mathcal{K}}u_{t,k}f(1) (\tilde\delta_{t,k}^{\mathrm{f}})^2 + r_{t,k}f(x_{t,k}+1)(\tilde\delta_{t,k}^{\mathrm{p}} )^2\\&\hspace{1mm}+\big[ 1-u_{t,k} f(1) - r_{t,k} f(x_{t,k}+1)    \big](\tilde\delta_{t,k})^2+\sum_{k\in \mathcal{K}}\sum_{k'\in \mathcal{K}, k'\neq k}u_{t,k}f(1)\tilde\delta_{t,k}^{\mathrm{f}}\tilde\delta_{t,k'}\\ &\hspace{1mm}+ r_{t,k}f(x_{t,k}+1)\tilde\delta_{t,k}^{\mathrm{p}} \tilde\delta_{t,k'}- u_{t,k} f(1)\tilde\delta_{t,k} \tilde\delta_{t,k'}-r_{t,k} f(x_{t,k}+1)\tilde\delta_{t,k} \tilde\delta_{t,k'}+ u_{t,k'}f(1)\tilde\delta_{t,k'}^{\mathrm{f}}\tilde\delta_{t,k}\\
   &\hspace{1mm}+ r_{t,k'}f(x_{t,k'}+1)\tilde\delta_{t,k'}^{\mathrm{p}} \tilde\delta_{t,k} -u_{t,k'} f(1)\tilde\delta_{t,k'} \tilde\delta_{t,k} -r_{t,k'} f(x_{t,k'}+1)\tilde\delta_{t,k'} \tilde\delta_{t,k} \\
   &\hspace{1mm}+\tilde\delta_{t,k'} \tilde\delta_{t,k}  +2KQ_{t} \sum_{k\in \mathcal{K}}u_{t,k} f(1)\tilde\delta_{t,k}^{\mathrm{f}}+ r_{t,k}f(x_{t,k}+1)\tilde\delta_{t,k}^{\mathrm{p}} \\
    &\hspace{1mm}+\big[ 1-u_{t,k} f(1) - r_{t,k} f(x_{t,k}+1)    \big]\tilde\delta_{t,k} \Big ] +\frac{1}{2} \big[({\Delta}^{\mathrm{max}})^2 - 2Q_{t} {\Delta}^{\mathrm{max}}\big ].
\end{array}
\end{equation*}
Accordingly, the solution of \eqref{lyp1} can be derived by solving the following problem 
\begin{spacing}{1.3}\vspace{-5mm}
\begin{subequations}
\label{fdp1}
    \begin{alignat}{2}
    \amin
    \quad & W_t\\
\mbox{subject to}\quad& \eqref{lyp1:2}-\eqref{lyp1:4},
    \end{alignat}
\end{subequations}
\end{spacing}
\noindent with variables ${\{u_{t,k},r_{t,k}\}_{k\in\mathcal{K}}}$.

The proposed LC-DT policy is summarized in Algorithm \ref{DCalg}. In Step 1, the transmitter updates the age of random arrival sources' fresh packets at slot $t$, $\delta_{t,k}^{\mathrm{f}}$, using \eqref{ageb}. In Step 2, given the current network state, it solves problem \eqref{fdp1} to find the optimal transmission decision at slot $t$. In Step 3, the network state (except $\delta_{t,k}^{\mathrm{f}}$ for all sources) is updated based on the current network state and the obtained transmission decision. 

Regarding finding the solution to problem \eqref{fdp1}, we observe that the number of feasible transmission decisions at each slot is a moderate number. Namely, the transmission options over the $K$ sources are sending a fresh packet (i.e., $K$ options), sending an under-process packet (i.e., $K$ options), or staying idle (i.e., $1$ option). On the other hand, if the under-process packet of a source was sent $x^{\mathrm{max}}$ times, the transmitter cannot send that packet. Thus, the number of feasible actions at each slot, i.e., the number of feasible combinations of variables $\{u_{t,k},r_{t,k}\}_{k\in \mathcal{K}}$ in problem \eqref{fdp1}, is at most $2K+1$. Because $2K+1$ is a small value for a reasonable system (increasing only linearly in $K$), we use the exhaustive search algorithm to solve problem \eqref{fdp1}.


\begin{algorithm}[t]
    \SetKwInOut{Input}{Initialize}
    \SetKwComment{Comment}{/* }{ */}
    \Input{Set $V$, and initialize $o_1$.} 

    \For {$t =1,2,3,\dots$}{{Step 1: Update $\delta_{t,k}^{\mathrm{f}}$ using \eqref{ageb}}\\
        {Step 2: Find decision variables ${\{u_{t,k},r_{t,k}}\}_{k\in\mathcal{K}}$ by solving problem \eqref{fdp1}}\\
        {Step 3: Update $x_{t+1,k}$ using \eqref{trr}, $\delta_{t+1,k}^{\mathrm{p}}$ using \eqref{ager}, $\delta_{t+1,k}^{\mathrm{}}$ using \eqref{age1}, and $Q_{t+1}$ using \eqref{vq1}}
      }    
    \caption{Proposed low-complexity dynamic transmission (LC-DT) policy}
    \label{DCalg}
\end{algorithm}

\subsection{Learning-based Transmission Policy: Unknown Environment}\label{unknown}
In this subsection, we assume that the environment is unknown, i.e., the transmitter does not know the system statistics, namely: 1) the packet arrival probability of each random arrival source, $\lambda_k$, and 2) the probability of successful decoding function, $f(\cdot)$. In this scenario, we transform CMDP problem \eqref{pcmdp01} into an unconstrained MDP problem via the Lyapunov DPP method. Then, we utilize the DQL algorithm \cite{dqn} to solve the MDP problem. Even though the DQL algorithm cannot ensure the optimality of the solution, the algorithm can be applied 1) without knowing the system statistics, 2) in a system with a large state and action space, and 3) without upper bounding AoI.

Inspired by \cite{wu2020} and using the results of Section \ref{known}, we transform CMDP problem \eqref{pcmdp01} into an MDP problem, in which our goal is to minimize the time average {DPP} function
\begin{equation}\label{pcmdp3}
\begin{array}{ll}
\underset{\pi}{\amin} \quad & 
\underset{T\rightarrow \infty}{\limisup} \frac{1}{T} \sum_{t=1}^{T}\mathbb{E}\left\{V\tau_t+ L(Q_{t+1})-L(Q_{t})\mid o_t\right\}.
\end{array}
\end{equation}

The system state of the MDP at slot $t$ is ${o_t = \big\{\{\delta_{t,k}^{\mathrm{f}}, \delta_{t,k}^{\mathrm{p}},\delta_{t,k},x_{t,k}\}_{k\in\mathcal{K}},Q_{t}\big\}}$ (defined in Section \ref{known}), and the action at time slot $t$ is $a_t=\{u_{t,k},r_{t,k}\}_{k\in\mathcal{K}}\in\mathcal{A}_{s_t}$ (defined in Section \ref{cmdpfm}). The cost function $\tilde{c}_t$, which is defined as the DPP function, is given as ${
\tilde c_t =V\tau_t +L(Q_{t+1}) -L(Q_{t})}$.
Note that the state transition probabilities of the MDP problem \eqref{pcmdp3} are unknown, as the environment is unknown.

To solve the MDP problem \eqref{pcmdp3}, we use the DQL algorithm in \cite[Algorithm 1]{dqn}. According to the DQL algorithm, an action is selected at each state to maximize a cumulative discounted immediate reward. As we aim to minimize cost function $\tilde c_t$, the immediate reward is defined as $r_t = -\tilde c_t$. 
The implementation of DQL, along with the parameters, is presented in Section \ref{nrds}.

\section{Numerical Results}\label{nrds}
In this section, we evaluate the performance of the proposed transmission scheduling policies, namely: 1) the  deterministic transmission policy presented in Algorithm \ref{Acmdp} in Section \ref{stationary}, 2) the low-complexity dynamic transmission (LC-DT) policy presented in Algorithm \ref{DCalg} in Section \ref{known}, and 3) the learning-based transmission policy presented in Section \ref{unknown}.

For the probability of successful decoding, we use the function in \cite{gz}, i.e.,  ${f(x_{t,k}) = {1-p_0\eta ^{x_{t,k}-1}}}$, where $p_0\in[0,1]$ is the error probability of the first transmission of a packet and $\eta\in[0,1]$ determines the effectiveness of the HARQ protocol. Unless otherwise specified, we consider one random arrival source and one generate-at-will source, i.e., $K=2$, and we set $\delta^{\mathrm{max}} = 18$ and $x^{\mathrm{max}} = 5$. The rest of the parameters are specified in each figure.

\subsection{Deterministic Transmission Policy}
In Fig.~\ref{fff1}, we evaluate the deterministic transmission policy {and the impact of system parameters on the performance}. For Algorithm~\ref{Acmdp}, we set the bounds on the Lagrangian multiplier as ${\beta_{\mathrm{u}} = 1}$, ${\beta_{\mathrm{l}} = 0}$, the bisection stopping criterion as ${\kappa = 0.005}$, and the RVIA stopping criterion as ${\epsilon = 0.01}$.

Fig.~\ref{fff1}(a) illustrates the evolution of average number of transmissions, $\bar\tau$, with respect to time slots for different values of maximum allowable average AoI, $\Delta^{\mathrm{max}}$, and the error probability of the first transmission, $p_0$. It can be seen from Fig.~\ref{fff1}(a) that $\bar\tau$ increases when $p_0$ increases. This behavior is due to the fact that when $p_0$ increases, the probability of successful decoding decreases, and thus, more transmission attempts are needed to meet the average AoI constraint. For example, when $\Delta^{\mathrm{max}} = 4$, by increasing $p_0$ from $0.4$ to $0.6$, $\bar\tau$ increases by about $50$ $\%$. In addition, $\bar\tau$ decreases when $\Delta^{\mathrm{max}}$ increases, because the transmitter needs fewer transmission attempts to satisfy the AoI constraint.

Fig.~\ref{fff1}(b) shows the evolution of $\bar\tau$ with respect to time slots for the different packet arrival rate, $\lambda_k$, and $\Delta^{\mathrm{max}}$.
From Fig.~\ref{fff1}(b), it can be seen that when $\lambda_k$ decreases, the average number of transmissions increases  dramatically. For example, when $\Delta^{\mathrm{max}} = 4$, by decreasing $\lambda_k$  from $0.5$ to $0.2$, the value of  $\bar\tau$ increases by about $75$ $\%$. This is because when $\lambda_k$ decreases, the availability of the fresh packets at the random arrival source decreases, and consequently, the AoI of this source increases. In this case, to satisfy the AoI constraint, the transmitter
must send the generate-at-will source's packets more frequently to compensate for the negative effect of the random arrival sources on the average AoI.


\begin{figure}[t]
\centering
\begin{subfigure}{.49\textwidth}
    \includegraphics[width = 8.1cm]{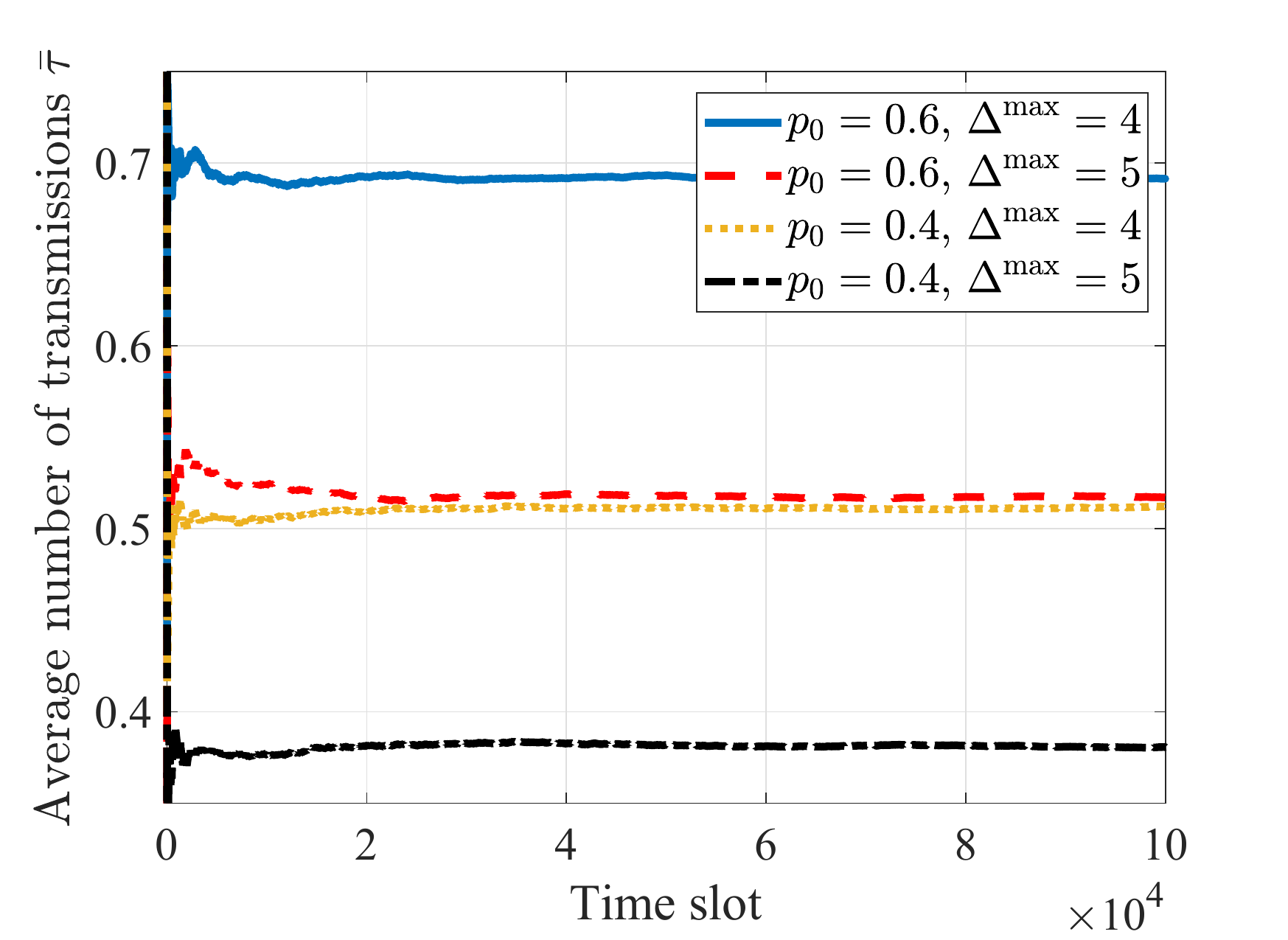}
    \caption{The evolution of $\bar\tau$ versus time slots for different $p_0$ with $\lambda_k = 0.7$ for all $k\in\mathcal{I}$.}
    \label{fanotcmdp}
    \end{subfigure}
\begin{subfigure}{.49\textwidth}
    \includegraphics[width = 8.1cm]{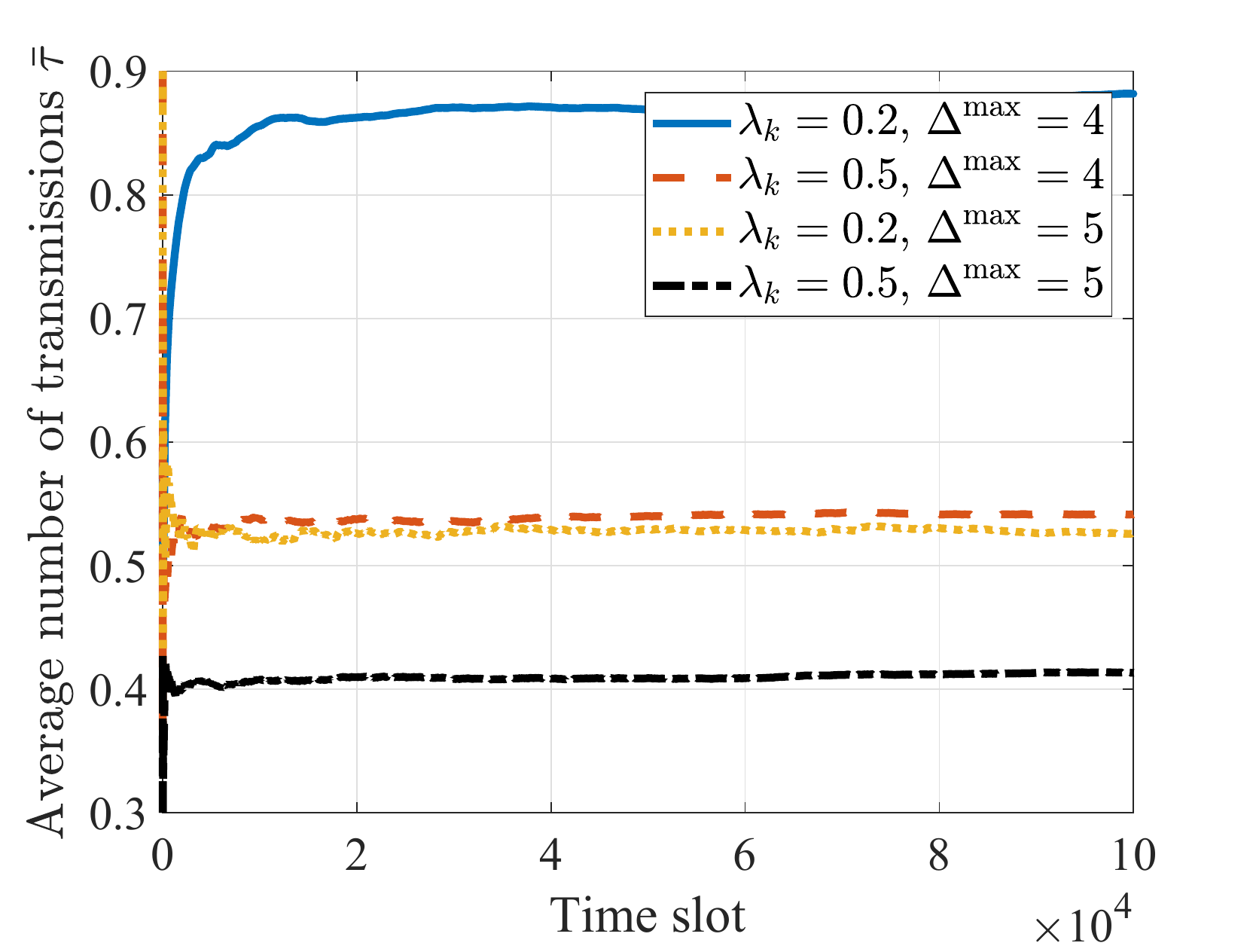}
    \caption{The evolution of $\bar\tau$ versus time slots for different $\lambda_k$ for $k\in\mathcal{I}$ with ${p_0 = 0.4}$.}
    \label{fanotcmdp2}
        \end{subfigure}\vspace{-2mm}
  \caption{The performance of the deterministic transmission policy for different $\Delta^{\mathrm{max}}$ with $\eta = 0.4$.}
  \label{fff1}\vspace{-10mm}
\end{figure}


\subsection{LC-DT Policy}
In this subsection, we evaluate the performance of the LC-DT policy.
Fig.~\ref{f3}(a) depicts the evolution of the average number of transmissions, $\bar\tau$, with respect to time slot for different values of the DPP trade-off parameter, $V$. Fig.~\ref{f3}(b) depicts the average AoI with respect to time slot for different values of $V$.

According to the figure, when $V$ increases, the average number of transmissions decreases and the average AoI increases. This is because by increasing $V$, the algorithm puts more emphasis on minimizing the average number of transmissions. In addition, Fig.~\ref{f3}(b) shows that for any value of $V$, the algorithm satisfies the constraint. Furthermore, for {$V\geq20$}, the obtained average AoI is about the maximum allowable average AoI; also, increasing $V$ beyond this value leads to negligible improvement in the objective function (see Fig.~\ref{f3}(a)). Therefore, it is beneficial to set $V$ larger than $30$ when utilizing the LC-DT policy.
\begin{figure}[t]
\centering
\begin{subfigure}{.49\textwidth}

\includegraphics[width = 8.1cm]{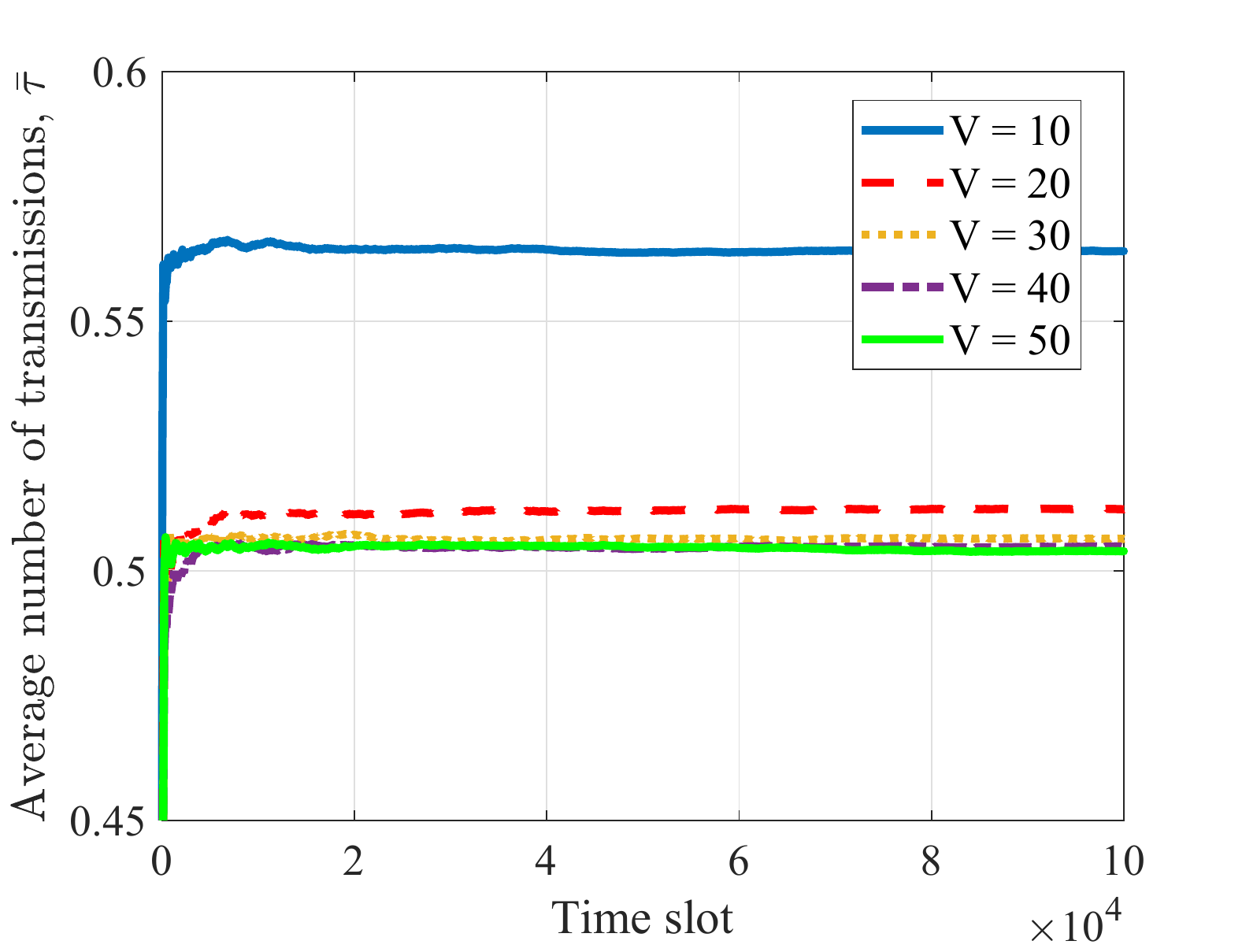}
\caption{The average number of transmissions.}
\label{f1}
\end{subfigure}
\begin{subfigure}{.49\textwidth}

\includegraphics[width = 8.3cm]{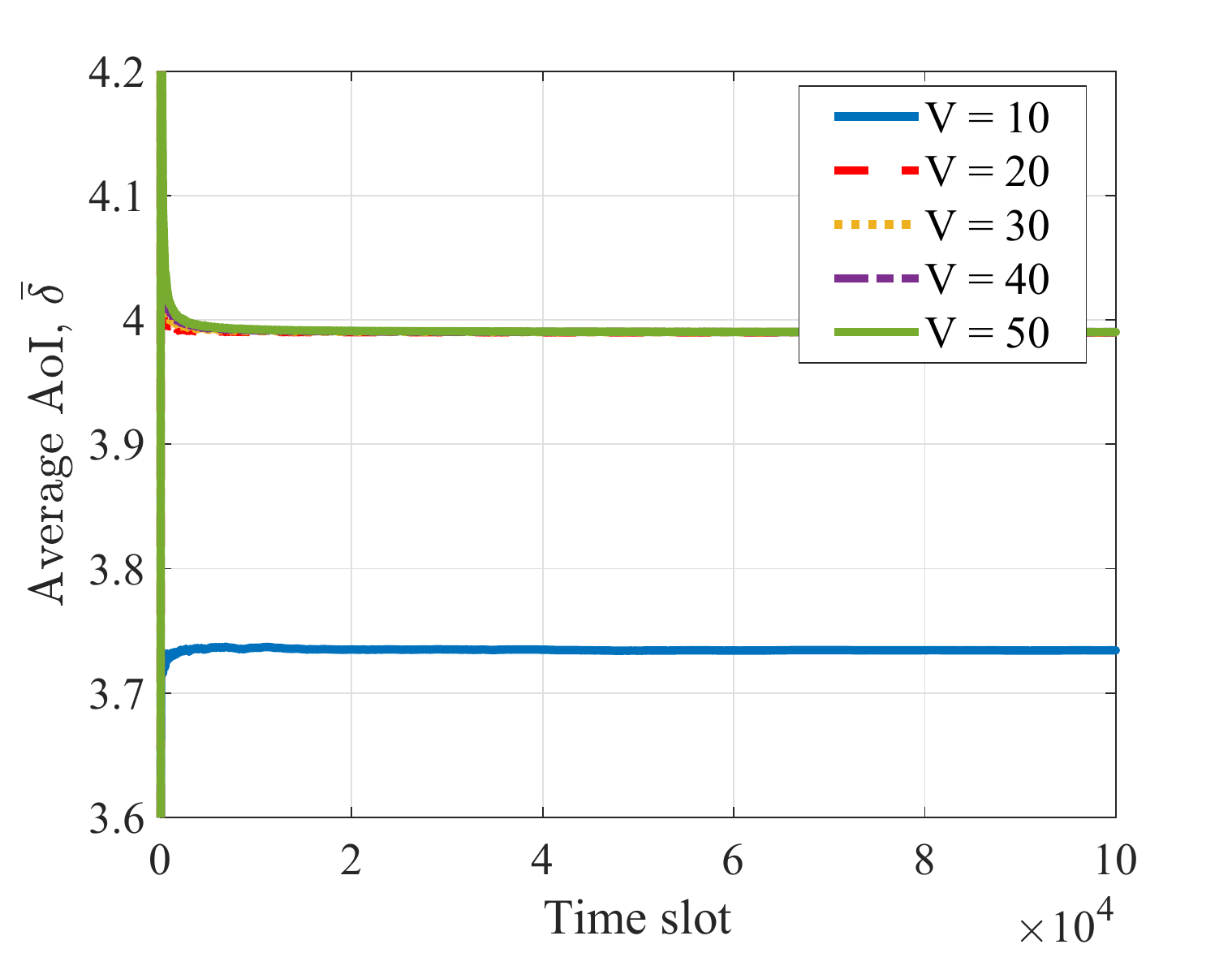}
\caption{The average AoI.}
\label{f2}
\end{subfigure}\vspace{-2mm}
\caption{The performance of the LC-DT policy versus time slots for different values of $V$, where $\eta = 0.4$, $p_0 = 0.4$, $\lambda_k = 0.7$ for all $k\in\mathcal{I}$, and $\Delta^{\mathrm{max}} = 4$.}
\label{f3}\vspace{-10mm}
\end{figure}

\subsection{Learning-based Transmission Policy}
In this section, we examine the learning-based transmission policy by evaluating the average number of transmissions, $\bar\tau$. To implement DQL algorithm, we use a fully-connected deep neural network with two hidden layers. Each hidden layer has $256$ neurons with \textit{ReLU} activation function. The optimizer is \textit{Adam}, the mini-batch size is $32$, and the replay memory size is $100000$. We set the learning rate as $0.001$, the discount factor as $0.99$, the number of steps per episode as $1000$, and the target network update rate as $\frac{1}{500}$. 
The convergence of the policy for different values of $V$ are shown in Fig.~\ref{flearning}(a). It can be seen that by taking about $800$ episodes, the DQL-agent (transmitter) is learned. After this, we stop the learning process and use the learned transmitter to operate in the system. Fig.~\ref{flearning}(b) shows the performance of the algorithm by evaluating $\bar\tau$ with respect to time slots for different values of $V$. Similar to Fig.~\ref{f3}, it can be seen that the performance is not considerably improved for $V$ larger than $30$.
\begin{figure}[t]
\centering
\begin{subfigure}{.49\textwidth}
    \centering
    \includegraphics[width = 8.2cm]{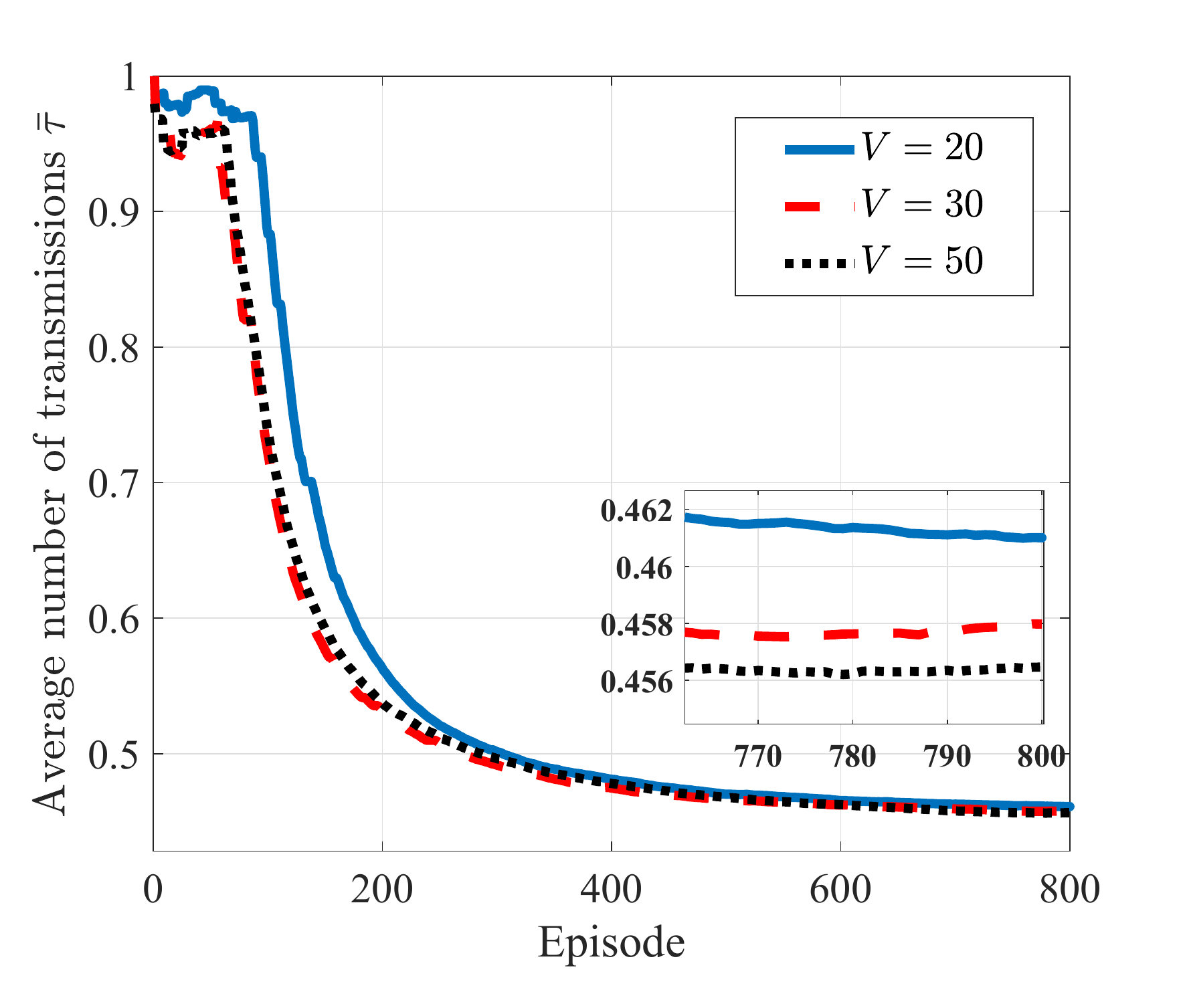}
    \caption{The evolution of $\bar\tau$ during the learning process.}
    \label{fepisode}
\end{subfigure}
\begin{subfigure}{.49\textwidth}
    \centering
    \includegraphics[width = 8.2cm]{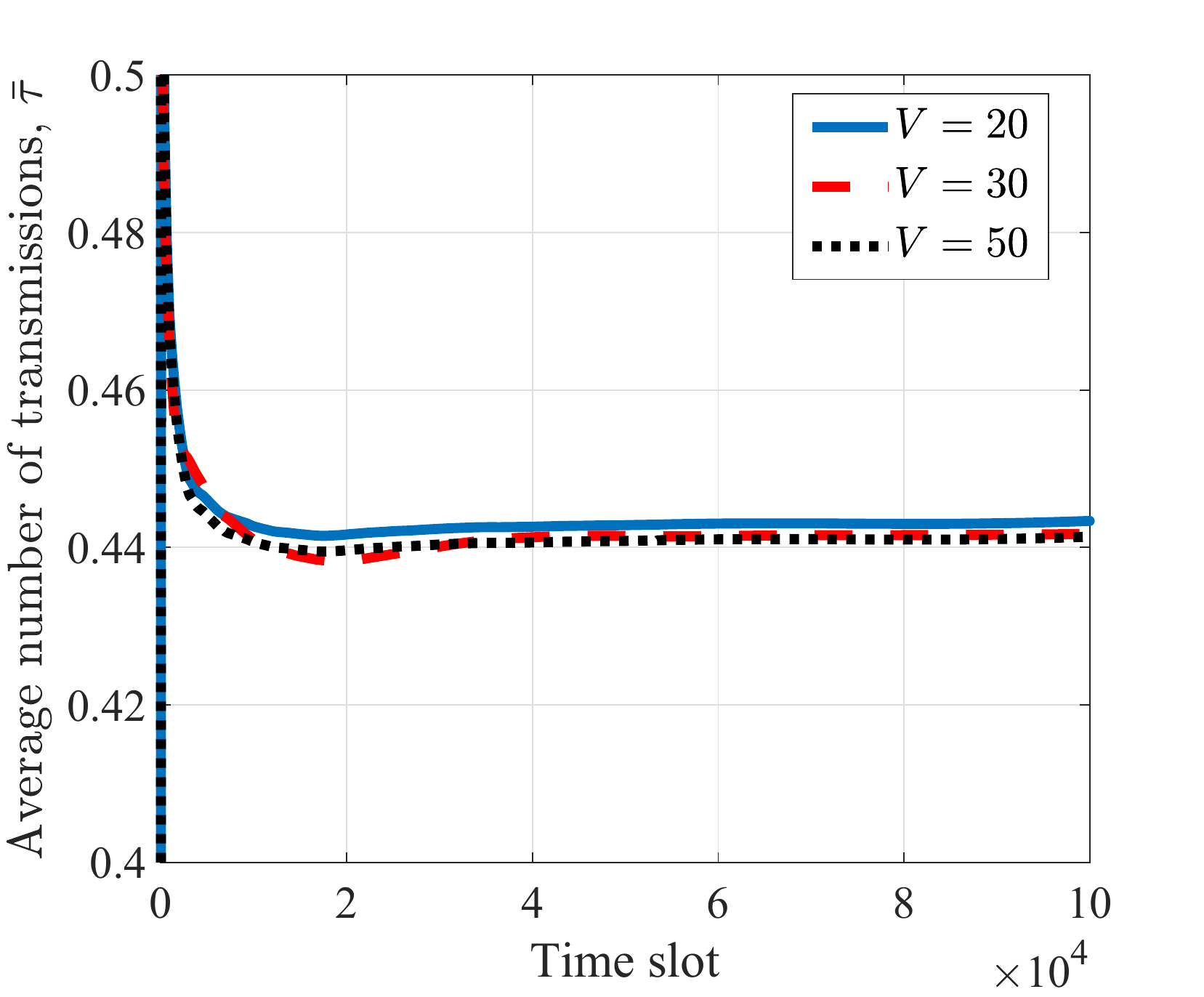}
    \caption{The evolution of $\bar\tau$ for the learned transmitter.}
    \label{flearningslot}
\end{subfigure}\vspace{-2mm}
\caption{The performance of the learning-based transmission policy for different values of $V$, where $\eta = 0.4$, $p_0 = 0.3$, $\lambda_k = 0.7$ for all $k\in\mathcal{I}$, and $\Delta^{\mathrm{max}} = 4$.}
\label{flearning}\vspace{-8mm}
\end{figure}

\subsection{Comparison of the  Proposed Policies}
Fig.~\ref{fcomp} shows the average number of transmissions, $\bar\tau$, as a function of  $\Delta^{\mathrm{max}}$ under different policies. In this figure, we plot the (infeasible) lower-bound policy, obtained in Algorithm \ref{Acmdp}, as a benchmark.
%
%
%
%
Furthermore, we consider a (feasible) baseline policy, where the transmitter sends a packet whenever the average AoI reaches $\Delta^{\text{max}}$. In every transmission attempt, the source with larger AoI is selected; if there are multiple sources with the largest AoI, one of them is selected randomly. The policy employs an HARQ protocol where the transmitter persistently re-transmits the packet at consecutive slots until it is transmitted successfully or reaches the maximum allowed number of transmissions $x^{\text{max}}$.
According to Fig.~\ref{fcomp}, as expected,
the lower-bound policy outperforms the other proposed policies as it does not satisfy the constraint. 
Both the LC-DT policy and the deterministic transmission policy have a small gap with the lower-bound policy, which implies their near-optimal performance. 
The learning-based transmission policy, which is obtained without knowing the system statistics, performs {relatively} close to the policies for the known environment, {especially for $\Delta^{\text{max}}\geq 4$.}
In general, compared to the baseline policy, the proposed policies improve the system performance considerably, e.g., the LC-DT policy provides about $40$ $\%$ improvement.

\begin{figure}[t]
    \centering
    \includegraphics[width = 9.0cm]{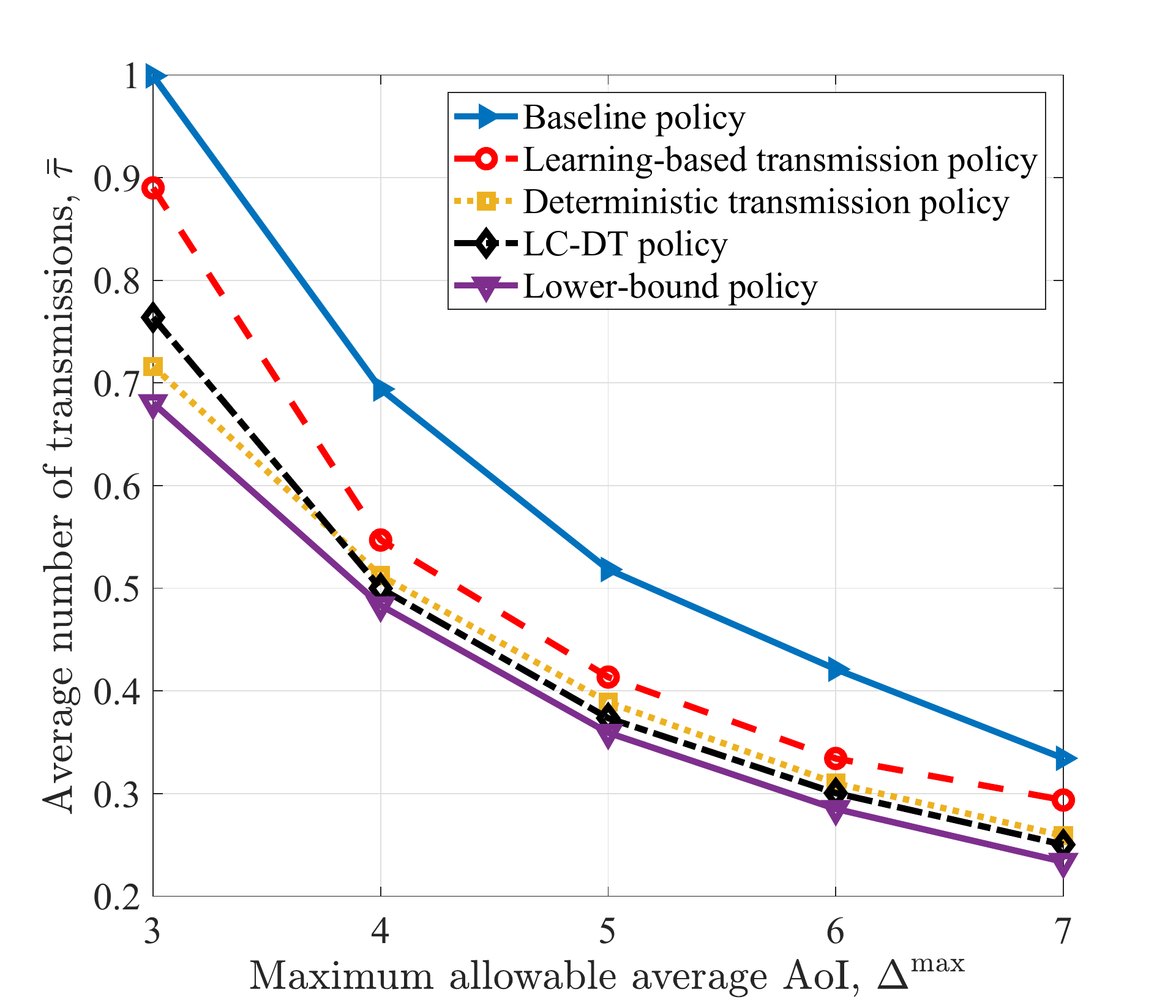}\vspace{-2mm}
    \caption{The average number of transmissions, $\bar\tau$, for the proposed transmission policies versus $\Delta^{\mathrm{max}}$ where $p_0 = 0.4$, $\eta = 0.4$, and $\lambda_k = 0.7$ for all $k\in\mathcal{I}$.}
    \label{fcomp}\vspace{-4mm}
\end{figure}

In Table~\ref{TC}, we compare the time complexity of the proposed transmission policies 
by evaluating the consumed time 1) in offline phase, i.e., initial processing time to find a policy and 2) in online phase, i.e., running time to find the optimal action at each slot. Regarding the offline phase, both the deterministic transmission policy and the learning-based transmission policy need to explore a large number of states, thus, they have a large initial processing time. Since the deterministic transmission policy needs to search for the optimal Lagrangian multiplier (bisection) and explore all the states (RVIA), its initial processing time is the longest. On the contrary, the LC-DT policy does not involve any initial processing (besides simply setting up problem \eqref{fdp1}).
In the online phase, the running time of the deterministic transmission policy is the smallest as the action selection is done through the lookup table. The running time of the LC-DT policy is larger, because it needs to solve optimization problem \eqref{fdp1}. The learning-based transmission policy has the largest running time, where a forward pass through the neural network is executed to select an action.

\begin{table}[t]
\caption{Time complexity of the proposed transmission scheduling policies}
    \centering
    \begin{tabular}{|l|l|l|}
    \hline
         Policy&Initial processing time (s)& Running time (ms) \\
         \hline
         Deterministic transmission policy&105013.745 & 0.031\\
         \hline
          LC-DT policy&0 & 0.342\\
         \hline
         Learning-based transmission policy&4029.38& 0.934\\
         \hline
    \end{tabular}
    \label{TC}
\end{table}

\subsection{Impact of HARQ}
In Fig.~\ref{harqdpp}, we study the impact of HARQ on the system's performance by evaluating the average number of transmissions, $\bar\tau$. Here, without loss of generality, we utilize the LC-DT policy.
Fig.~\ref{harqdpp}(a) shows the effect of the maximum allowed number of transmissions for a packet, $x^{\mathrm{max}}$, where $x^{\mathrm{max}}=1$ indicates that the system operates without HARQ. As it can be seen, for high probability of successful decoding (small $p_0$), i.e., for good channel conditions, HARQ is not that beneficial. This is due to the fact that with high probability of successful decoding, most of the packets are successfully decoded in the first transmission attempt. However, in bad channel conditions, HARQ plays an important role. 
For example, when $p_0 = 0.6$, the performance improves substantially by increasing $x^{\mathrm{max}}$ from $1$ to $2$, i.e., by merely activating the HARQ with only one allowed retransmission. It is worth noting that when $p_0 = 0.7$, the transmitter cannot satisfy the average AoI constraint without HARQ ($x^{\mathrm{max}}=1$). This is because when $p_0$ is large, the first transmit attempts tend to fail, in which case the retransmissions would be the crucial enabler for successful receptions of the packets. 

\begin{figure}[t]
\centering
\begin{subfigure}{.49\textwidth}
    \centering
    \includegraphics[width = 8.2cm]{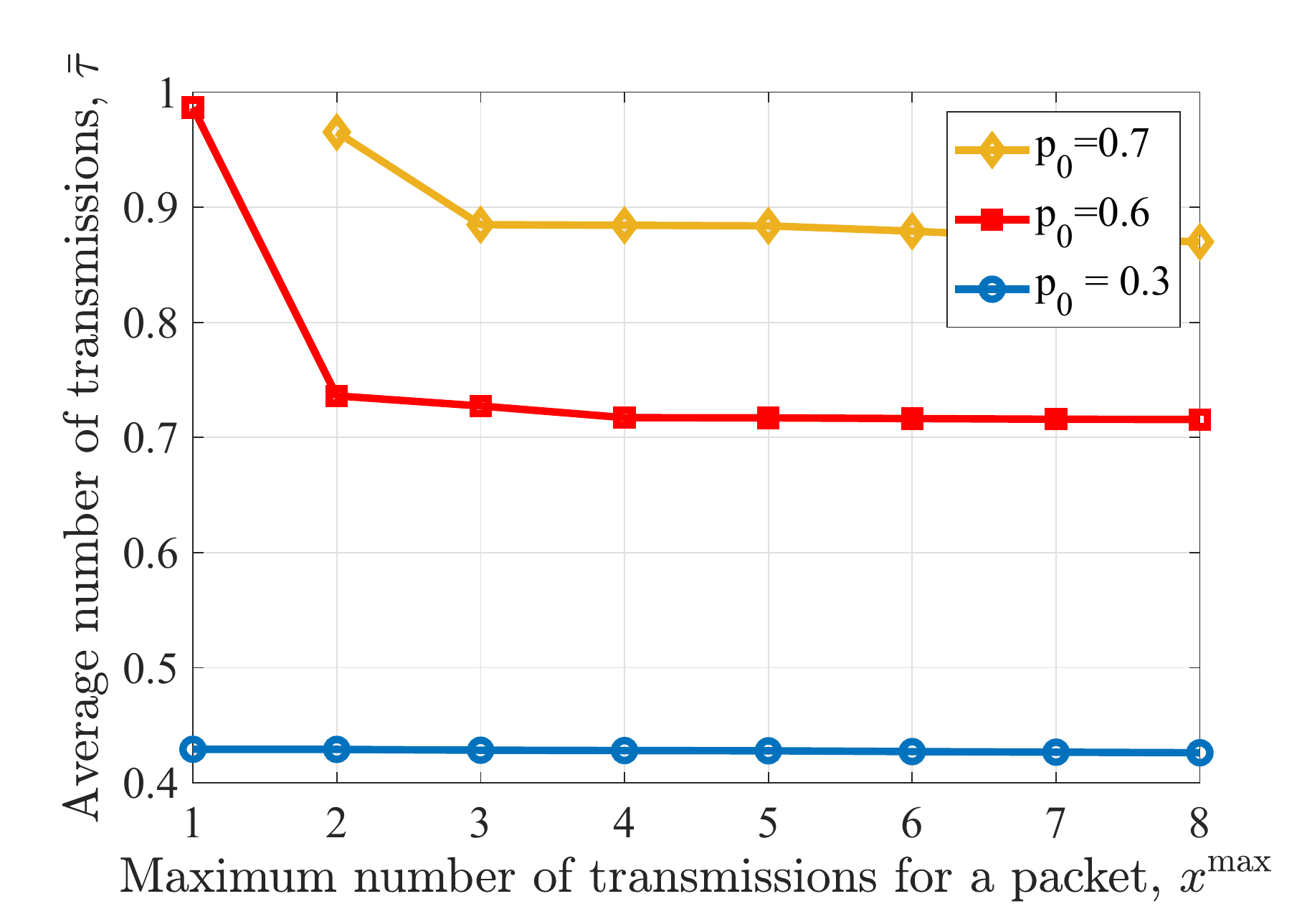}
    \caption{The average number of transmissions, $\bar\tau$, versus $x^{\mathrm{max}}$, where ${\eta = 0.3}$ and $\Delta^{\mathrm{max}} = 4$.}
    \label{f51}
    \end{subfigure}
\begin{subfigure}{.49\textwidth}
    \includegraphics[width = 8.2cm]{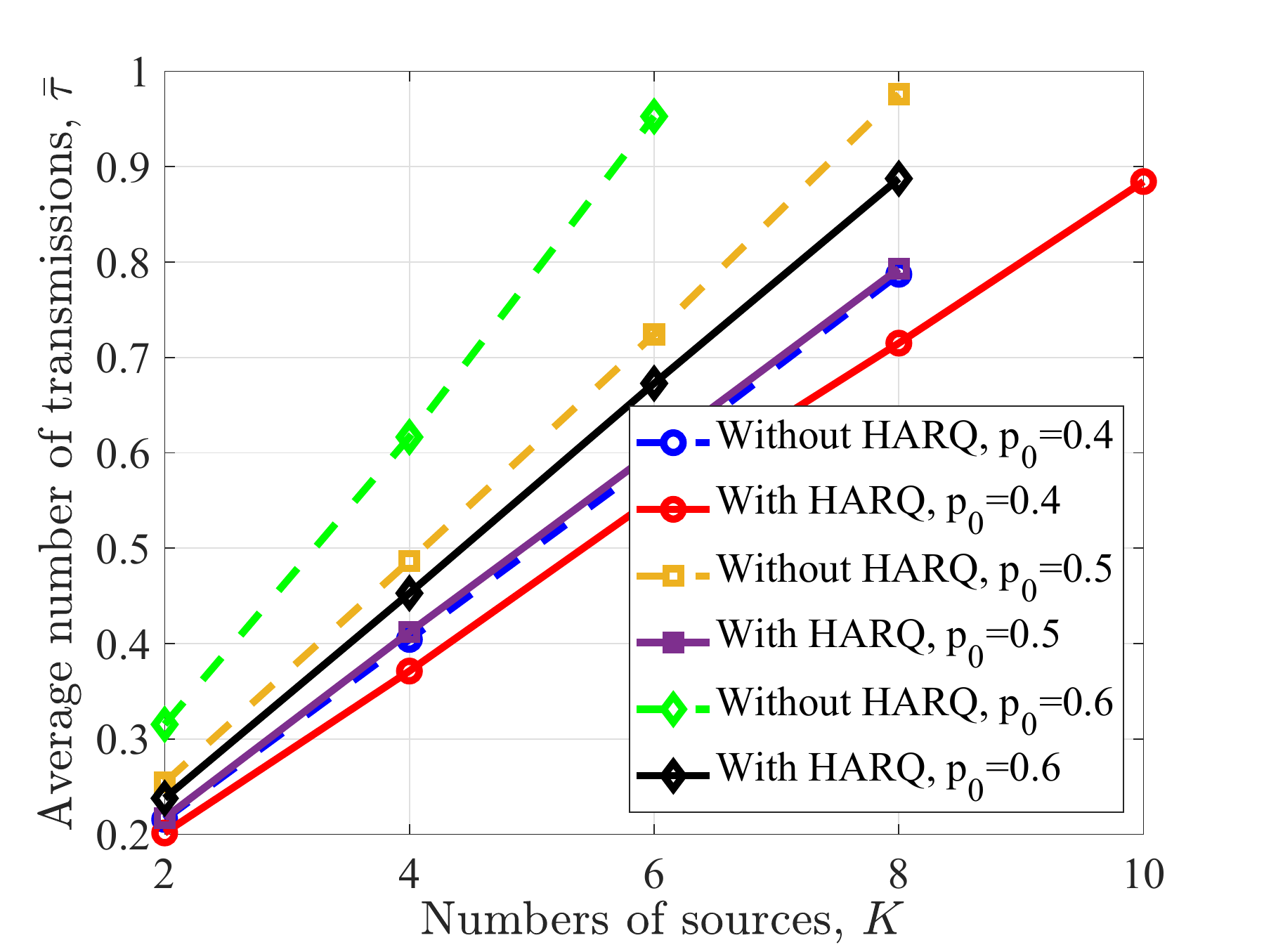}
    \caption{The average number of transmissions, $\bar\tau$, versus $K$, where ${\eta = 0.4}$, $x^{\mathrm{max}}=5$, and $\Delta^{\mathrm{max}} = 10$.}
    \label{f5}
    \end{subfigure}\vspace{-2mm}
    \caption{The impact of employing HARQ on the performance of the system for different $p_0$.}
    \label{harqdpp}\vspace{-8mm}
\end{figure}

In Fig.~\ref{harqdpp}(b), we study the effect of employing HARQ by evaluating $\bar\tau$ with respect to the number of sources $K$, where $K = I+J$ and $I=J$. According to the figure, $\bar\tau$ increases by increasing $K$. 
This is due to the fact that because the transmitter has to send each source's packets regularly to satisfy the AoI constraint, the increased number of sources inevitably leads to more transmissions in the system. Moreover, similar to Fig.~\ref{harqdpp}(a), it can be seen in Fig.~\ref{harqdpp}(b) that employing HARQ decreases $\bar\tau$, as HARQ benefits from the increased probability of successful decoding via retransmissions. For example, with $p_0=0.6$ and $K=6$, $\bar\tau$ with HARQ is $30$ $\%$ less than in the case without HARQ. Interestingly, the system cannot support $K=10$ sources for $p_0\geq 0.4$, unless HARQ is employed.

\section{Conclusion}\label{clcn}
We studied an HARQ-based multi-source status update system with random arrival and generate-at-will sources, communicating through an error-prone channel. We solved the problem of minimizing the average number of transmissions subject to the average AoI constraint in the known and unknown environments. We developed a deterministic transmission policy using the RVIA and the bisection for the known environment. For the sake of reduced computational complexity, we developed the LC-DT policy using the DPP method for the known environment. For the unknown environment, we utilized the DPP method and the DQL algorithm to develop a learning-based transmission policy. The numerical results showed the near-optimal performance of the deterministic transmission policy and the LC-DT policy. Also, the learning-based policy attained performance relatively close to the policies developed for the known environment. 
Overall, the results showed about $40$ $\%$ performance gain for the proposed policies over a baseline scheduling policy and demonstrated the great potential of HARQ to improve information freshness in multi-source status update systems.

\appendices

\section{Proof of Lemma \ref{lm1}}\label{plm1}
To derive $\mathbb{E}\{\hat\delta_{t+1} \mid o_t\}$, we use the definition in \eqref{avokaoi}, $\hat\delta_{t} = {\frac{1}{K}\sum_{k=1}^{K} \delta_{t,k}}$,
and re-express $\delta_{t+1,k}$ via \eqref{age1} as
\begin{equation}\label{ages1}
\begin{array}{ll}
    \delta_{t+1,k}& \hspace{-3mm}= u_{t,k}d_{t}\tilde\delta_{t,k}^{\mathrm{f}}+ r_{t,k}d_{t}\tilde\delta_{t,k}^{\mathrm{p}} + \big[ u_{t,k}(1-d_{t})+ r_{t,k}(1-d_{t}) + (1-u_{t,k}-r_{t,k})  \big]\tilde\delta_{t,k}\\
    & \hspace{-3mm}= u_{t,k}d_{t}\tilde\delta_{t,k}^{\mathrm{f}}+ r_{t,k}d_{t}\tilde\delta_{t,k}^{\mathrm{p}} + \big[ u_{t,k}-u_{t,k}d_{t}+ r_{t,k}-r_{t,k}d_{t} + 1-u_{t,k}-r_{t,k}  \big]\tilde\delta_{t,k}\\
    & \hspace{-3mm}= u_{t,k}d_{t}\tilde\delta_{t,k}^{\mathrm{f}}+ r_{t,k}d_{t}\tilde\delta_{t,k}^{\mathrm{p}}+ \big[ 1- d_{t}(u_{t,k}+r_{t,k}) \big]\tilde\delta_{t,k}.
\end{array}
\end{equation}
Taking conditional expectation in \eqref{ages1}, $\mathbb{E}\big\{\delta_{t + 1,k}\mid o_t\big\}$ is expressed as
\begin{equation}\label{pureeaoi}
\begin{array}{ll}
   \!\!\!\!\mathbb{E}\{\delta_{t+1,k} \mid o_t\} & \hspace{-3mm}= \mathbb{E}\{u_{t,k}d_{t}\tilde\delta_{t,k}^{\mathrm{f}}\mid o_t\} + \mathbb{E}\{r_{t,k}d_{t}\tilde\delta_{t,k}^{\mathrm{p}}\mid o_t\}+ \mathbb{E}\Big\{\big[ 1- d_{t}(u_{t,k}+r_{t,k}) \big]\tilde\delta_{t,k}\mid o_t\Big\}\\
    &\hspace{-3mm}\overset{(a)}{=} \mathbb{E}\{u_{t,k}d_{t}\mid o_t\}\tilde\delta_{t,k}^{\mathrm{f}} + \mathbb{E}\{r_{t,k}d_{t}\mid o_t\}\tilde\delta_{t,k}^{\mathrm{p}}+\big[ 1- \mathbb{E}\{d_{t}(u_{t,k}+r_{t,k})\mid o_t\} \big]\tilde\delta_{t,k},
    \end{array}
\end{equation}
where equality $(a)$ follows from the fact that $\delta_{t,k}^{\mathrm{f}}$, $\delta_{t,k}^{\mathrm{p}}$, and $\delta_{t,k}$ are given by the network state.

We need to calculate $\mathbb{E}\{u_{t,k}d_{t}\mid o_t\}$ and $\mathbb{E}\{r_{t,k}d_{t}\mid o_t\}$ in \eqref{pureeaoi} which are given by the following two lemmas.

\begin{lemma}\label{l1}
For any source $k$, the conditional expectation $\mathbb{E}\{u_{t,k}d_{t}\mid o_t\}$ is given as
\begin{equation}\label{eexpages2}
\begin{array}{ll}
&\mathbb{E}\{u_{t,k}d_{t}\mid o_t\} = \mathbb{E}\{u_{t,k}\mid o_t\}f(1).
\end{array}
\end{equation}
\end{lemma}
\begin{proof}
Based on the law of iterated expectations, we have
\begin{equation*}\label{eexpages123}
\begin{array}{ll}
   & \mathbb{E}\{u_{t,k}d_{t}\mid o_t\}= \mathbb{E}\big\{\mathbb{E}\{u_{t,k}d_{t}\mid o_t,u_{t,k}\}\big\} =\mathbb{E}\{1d_{t}\mid o_t,u_{t,k}=1\}\mathrm{Pr}(u_{t,k}=1\mid o_t) \\
    &\hspace{1mm}+\mathbb{E}\{0d_t\mid o_t,u_{t,k}=0\}\mathrm{Pr}(u_{t,k}=0\mid o_t) =\Big(1f(1) + 0\big(1-f(1)\big)\Big)\mathrm{Pr}(u_{t,k}=1\mid o_t)\\&\hspace{-0mm}= f(1)\mathrm{Pr}(u_{t,k}=1\mid o_t)\overset{(a)}{=}\mathbb{E}\{u_{t,k}\mid o_t\}f(1),
    \end{array}
\end{equation*}
where the equality ($a$) comes from the following equality
\begin{equation}\label{ap4}
    \begin{array}{ll}
\mathbb{E}\{u_{t,k}\mid o_t\} = 0\mathrm{Pr}(u_{t,k}=0\mid o_t) + 1\mathrm{Pr}(u_{t,k}=1\mid o_t) = \mathrm{Pr}(u_{t,k}=1\mid o_t).
    \end{array}
\end{equation}
\end{proof}
\begin{lemma}\label{l2}
For any source $k$, the conditional expectation $\mathbb{E}\{r_{t,k}d_{t}\mid o_t\}$ is given as
\begin{equation}\label{eexpages3}
    \begin{array}{ll}
    &\mathbb{E}\{r_{t,k}d_{t}\mid o_t\} = \mathbb{E}\{r_{t,k}\mid o_t\}f(x_{t,k}+1).\\
    \end{array}
\end{equation}
\end{lemma}
\begin{proof}
Following the same steps as in the proof of Lemma \ref{l1}, we have
\begin{equation*}\label{eexpages1}
\begin{array}{ll}
    &\mathbb{E}\{r_{t,k}d_{t}\mid o_t\}  = \mathbb{E}\big\{\mathbb{E}\{r_{t,k}d_{t}\mid o_t,r_{t,k}\}\big\} =\mathbb{E}\{1d_{t}\mid o_t,r_{t,k}=1\}\mathrm{Pr}(r_{t,k}=1\mid o_t)\\&+ \mathbb{E}\{0d_{t}\mid o_t,r_{t,k}=0\}\mathrm{Pr}(r_{t,k}=0\mid o_t) \\&=\Big(1f(x_{t,k}+1) + 0\big(1-f(x_{t,k}+1)\big)\Big)\mathrm{Pr}(r_{t,k}=1\mid o_t)\\
    &\hspace{0mm}= f(x_{t,k}+1)\mathrm{Pr}(r_{t,k}=1\mid o_t)\overset{(a)}{=} \mathbb{E}\{r_{t,k}\mid o_t\}f(x_{t,k}+1),
    \end{array}
\end{equation*}
where the equality ($a$) comes from the following equality
\begin{equation}\label{ap5}
    \begin{array}{ll}
\mathbb{E}\{r_{t,k}\mid o_t\} = 0\mathrm{Pr}(r_{t,k}=0\mid o_t) + 1\mathrm{Pr}(r_{t,k}=1) = \mathrm{Pr}(r_{t,k}=1\mid o_t).
    \end{array}
\end{equation}
\end{proof}
Using Lemmas \ref{l1} and \ref{l2}, the expression in \eqref{pureeaoi} becomes
\begin{equation}\label{eaoi2}
    \begin{array}{ll}
    
    \mathbb{E}\big\{\hat\delta_{t + 1}\mid o_t\big\}& \hspace{-3mm}=\frac{1}{K}\sum_{k\in\mathcal{K}}\mathbb{E}\{u_{t,k}\mid o_t\}f(1)\tilde\delta_{t,k}^{\mathrm{f}} + \mathbb{E}\{r_{t,k}\mid o_t\}f(x_{t,k}+1)\tilde\delta_{t,k}^{\mathrm{p}}\\
    &\hspace{1mm}+ \big[1-f(1)\mathbb{E}\{u_{t,k}\mid o_t\}-f(x_{t,k}+1)\mathbb{E}\{r_{t,k}\mid o_t\}\big]\tilde\delta_{t,k}.
    
    \end{array}
\end{equation}

\section{Proof of Lemma \ref{lm2}}\label{plm2}
To derive $\mathbb{E}\{\hat{\delta}^2_{t + 1}\mid o_t\}$, we use the definition in \eqref{avokaoi}, $\hat\delta_{t} = {\frac{1}{K}\sum_{k=1}^{K} \delta_{t,k}}$, and re-express it as
\begin{equation}\label{mainl2}
    \begin{array}{ll}
        \mathbb{E}\{\hat{\delta}^2_{t + 1}\mid o_t\} &\hspace{-3mm}= \frac{1}{K^2}\mathbb{E}\{(\delta_{t + 1,1}+\dots +\delta_{t + 1,K})^2\mid o_t\}\\
 &\hspace{-3mm}= \frac{1}{K^2}\big[\sum_{k\in\mathcal{K}}\mathbb{E}\{\delta_{t + 1,k}^2\mid o_t\} + \sum_{k\in\mathcal{K}}\sum_{k'\in\mathcal{K},k'\neq k}\mathbb{E}\{\delta_{t + 1,k}\delta_{t + 1,k'}\mid o_t\}\big].
    \end{array}
\end{equation}
We need to calculate the terms $\mathbb{E}\{\delta_{t + 1,k}^2\mid o_t\}$ and $\mathbb{E}\{\delta_{t + 1,k}\delta_{t + 1,k'}\mid o_t\}$. As the conditions in \eqref{age1} are mutually exclusive and collectively exhaustive and the square of a binary variable equals the variable itself, $\delta^2_{t + 1,k}$ can be calculated from \eqref{age1} and \eqref{ages1} as 
\begin{equation}\label{ages2}
    \begin{array}{ll}
    
    \delta_{t+1,k}^2& \hspace{-3mm}= u_{t,k}d_{t}(\tilde\delta_{t,k}^{\mathrm{f}})^2 + r_{t,k}d_{t}(\tilde\delta_{t,k}^{\mathrm{p}})^2 + \big[ 1- d_{t}(u_{t,k}+r_{t,k}) \big](\tilde\delta_{t,k})^2.
    
    \end{array}
\end{equation}
Taking conditional expectation in \eqref{ages2}, $\mathbb{E}\big\{\delta^2_{t + 1,k}\mid o_t\big\}$ is expressed as
\begin{equation}\label{pureeaoi1}
\begin{array}{ll}
   \!\!\mathbb{E}\{\delta_{t+1,k}^2 \mid o_t\}  \hspace{0mm}= \mathbb{E}\{u_{t,k}d_{t}(\tilde\delta_{t,k}^{\mathrm{f}} )^2\mid o_t\} + \mathbb{E}\{r_{t,k}d_{t}(\tilde\delta_{t,k}^{\mathrm{p}} )^2\mid o_t\}\\
    \hspace{0mm}+ \mathbb{E}\big\{\big[ 1- d_{t}(u_{t,k}+r_{t,k}) \big](\tilde\delta_{t,k})^2\mid o_t\big\}\overset{(a)}{=} \mathbb{E}\{u_{t,k}d_{t}\mid o_t\}(\tilde\delta_{t,k}^{\mathrm{f}} )^2 + \mathbb{E}\{r_{t,k}d_{t}\mid o_t\}(\tilde\delta_{t,k}^{\mathrm{p}} )^2\\
    \hspace{0mm}+\big[ 1- \mathbb{E}\{d_{t}(u_{t,k}+r_{t,k})\mid o_t\} \big](\tilde\delta_{t,k} )^2,
    \end{array}
\end{equation}
where equality $(a)$ follows from the fact that $\delta_{t,k}^{\mathrm{f}}$, $\delta_{t,k}^{\mathrm{p}}$, and $\delta_{t,k}$ are given by the network state.

Using Lemmas \ref{l1} and \ref{l2}, the expression in \eqref{pureeaoi1} is calculated as 
\begin{equation}\label{pureeaoi12}
\begin{array}{ll}
   \mathbb{E}\{\delta_{t+1,k}^2 \mid o_t\} & \hspace{-3mm}=  \mathbb{E}\{u_{t,k}\mid o_t\}f(1)(\tilde\delta_{t,k}^{\mathrm{f}})^2 + \mathbb{E}\{r_{t,k}\mid o_t\}f(x_{t,k}+1)(\tilde\delta_{t,k}^{\mathrm{p}})^2\\
    &\hspace{1mm}+\big[1-f(1)\mathbb{E}\{u_{t,k}\mid o_t\}-f(x_{t,k}+1)\mathbb{E}\{r_{t,k}\mid o_t\}\big](\tilde\delta_{t,k} )^2.
    \end{array}
\end{equation}

Now, we calculate $\delta_{t + 1,k}{\delta}_{t + 1,k'}$ in expression \eqref{mainl2}. Based on \eqref{ages1}, we can express it as 
\begin{equation}\label{ages3ap1}
    \begin{array}{ll}
    \!\!\!\!\delta_{t+1,k}\delta_{t+1,k'}&\hspace{-3mm}= u_{t,k}d_{t}\big(1-d_{t}(u_{t,k'}+r_{t,k'})\big)\tilde\delta_{t,k}^{\mathrm{f}} \tilde\delta_{t,k'} + r_{t,k}d_{t}\big(1-d_{t}(u_{t,k'}+r_{t,k'})\big)\tilde\delta_{t,k}^{\mathrm{p}} \tilde\delta_{t,k'} \\
    &\hspace{1mm}+\big(1-d_{t}(u_{t,k}+r_{t,k})\big)(d_{t}u_{t,k'})\tilde\delta_{t,k}\tilde\delta_{t,k'}^{\mathrm{f}}+\big(1-d_{t}(u_{t,k}+r_{t,k})\big)(d_{t}r_{t,k'})\tilde\delta_{t,k}\tilde\delta_{t,k'}^{\mathrm{p}}\\
&\hspace{1mm}+\big(1-d_{t}(u_{t,k}+r_{t,k})\big)\big(1-d_{t}(u_{t,k'}+r_{t,k'})\big)\tilde\delta_{t,k'} \tilde\delta_{t,k}.
    \end{array}
\end{equation}
Because the transmitter can transmit one packet per slot, we have $u_{t,k}u_{t,k'}=0$ and ${u_{t,k}r_{t,k'}=0}$. Thus, the expression in \eqref{ages3ap1} is rewritten as
\begin{equation}\label{ages3ap2}
    \begin{array}{ll}
    \delta_{t+1,k}\delta_{t+1,k'}&\hspace{-3mm}=u_{t,k}d_{t}\tilde\delta_{t,k}^{\mathrm{f}} \tilde\delta_{t,k'} + r_{t,k}d_{t}\tilde\delta_{t,k}^{\mathrm{p}} \tilde\delta_{t,k'}+ u_{t,k'}d_{t}\tilde\delta_{t,k'}^{\mathrm{f}} \tilde\delta_{t,k}  + r_{t,k'}d_{t}\tilde\delta_{t,k'}^{\mathrm{p}} \tilde\delta_{t,k} \\
    &\hspace{1mm}-u_{t,k}d_{t}\tilde\delta_{t,k} \tilde\delta_{t,k'} -r_{t,k}d_{t}\tilde\delta_{t,k} \tilde\delta_{t,k'} -u_{t,k'}d_{t}\tilde\delta_{t,k'} \tilde\delta_{t,k} -r_{t,k'}d_{t}\tilde\delta_{t,k'} \tilde\delta_{t,k} +\tilde\delta_{t,k'} \tilde\delta_{t,k}.
\end{array}
\end{equation}
Using Lemmas \ref{l1} and \ref{l2}, the conditional expectation of the expression in \eqref{ages3ap2} is calculated as 
\begin{equation}\label{emtaoi3}
    \begin{array}{ll}
    \!\!\!\mathbb{E}\big\{\delta_{t+1,k}\delta_{t+1,k'}\mid o_t\big\}=\mathbb{E}\{u_{t,k}\mid o_t\}f(1)\tilde\delta_{t,k}^{\mathrm{f}}\tilde\delta_{t,k'} + \mathbb{E}\{r_{t,k}\mid o_t\}f(x_{t,k}+1)\tilde\delta_{t,k}^{\mathrm{p}} \tilde\delta_{t,k'} \\ + \mathbb{E}\{u_{t,k'}\mid o_t\}f(1)\tilde\delta_{t,k'}^{\mathrm{f}}\tilde\delta_{t,k} + \mathbb{E}\{r_{t,k'}\mid o_t\}f(x_{t,k'}+1)\tilde\delta_{t,k'}^{\mathrm{p}} \tilde\delta_{t,k} \\- \mathbb{E}\{u_{t,k}\mid o_t\} f(1)\tilde\delta_{t,k} \tilde\delta_{t,k'} -\mathbb{E}\{r_{t,k}\mid o_t\} f(x_{t,k}+1)\tilde\delta_{t,k} \tilde\delta_{t,k'} \\-\mathbb{E}\{u_{t,k'}\mid o_t\} f(1)\tilde\delta_{t,k'} \tilde\delta_{t,k} -\mathbb{E}\{r_{t,k'}\mid o_t\} f(x_{t,k'}+1)\tilde\delta_{t,k'} \tilde\delta_{t,k} +\tilde\delta_{t,k'} \tilde\delta_{t,k}. 
    \end{array}
\end{equation}

Substituting \eqref{pureeaoi12} and \eqref{emtaoi3} into \eqref{mainl2}, we  derive
\begin{equation*}\label{mainl21}
    \begin{array}{ll}
        \mathbb{E}\{\hat{\delta}^2_{t + 1}\mid o_t\} 
 \hspace{0mm}= \frac{1}{K^2}\Big[\sum_{k\in\mathcal{K}}\mathbb{E}\{u_{t,k}\mid o_t\}f(1)(\tilde\delta_{t,k}^{\mathrm{f}})^2 + \mathbb{E}\{r_{t,k}\mid o_t\}f(x_{t,k}+1)(\tilde\delta_{t,k}^{\mathrm{p}})^2\\
    \hspace{1mm}+\big[1-f(1)\mathbb{E}\{u_{t,k}\mid o_t\}-f(x_{t,k}+1)\mathbb{E}\{r_{t,k}\mid o_t\}\big](\tilde\delta_{t,k} )^2\\ \hspace{1mm} + \sum_{k\in\mathcal{K}}\sum_{k'\in\mathcal{K},k'\neq k}\mathbb{E}\{u_{t,k}\mid o_t\}f(1)\tilde\delta_{t,k}^{\mathrm{f}}\tilde\delta_{t,k'} + \mathbb{E}\{r_{t,k}\mid o_t\}f(x_{t,k}+1)\tilde\delta_{t,k}^{\mathrm{p}} \tilde\delta_{t,k'} \\ \hspace{1mm}+ \mathbb{E}\{u_{t,k'}\mid o_t\}f(1)\tilde\delta_{t,k'}^{\mathrm{f}}\tilde\delta_{t,k} + \mathbb{E}\{r_{t,k'}\mid o_t\}f(x_{t,k'}+1)\tilde\delta_{t,k'}^{\mathrm{p}} \tilde\delta_{t,k} \\ \hspace{1mm}- \mathbb{E}\{u_{t,k}\mid o_t\} f(1)\tilde\delta_{t,k} \tilde\delta_{t,k'} -\mathbb{E}\{r_{t,k}\mid o_t\} f(x_{t,k}+1)\tilde\delta_{t,k} \tilde\delta_{t,k'} \\
    \hspace{1mm}-\mathbb{E}\{u_{t,k'}\mid o_t\} f(1)\tilde\delta_{t,k'} \tilde\delta_{t,k} -\mathbb{E}\{r_{t,k'}\mid o_t\} f(x_{t,k'}+1)\tilde\delta_{t,k'} \tilde\delta_{t,k} +\tilde\delta_{t,k'} \tilde\delta_{t,k}\Big].
    \end{array}
\end{equation*}

\section{The expression for the objective function of problem \eqref{fdp1}}\label{pfp}
To derive $W_t$, we rewrite \eqref{eqQ} using Lemmas \ref{lm1} and \ref{lm2} as 
\begin{equation}\label{eqQa}
\begin{array}{ll}
   \hspace{-3mm}V\mathbb{E}\{\tau_{t}\mid o_t \} + \alpha(o_{t})\hspace{-0mm}\leq V\sum_{k\in \mathcal{K}}\mathbb{E}\{u_{t,k}\mid o_t\}+ \mathbb{E}\{r_{t,k}\mid o_t\}+\frac{1}{2} \big(({\Delta}^{\mathrm{max}})^2 \\
   \hspace{1mm} + \frac{1}{K^2}\Big[\sum_{k\in\mathcal{K}}\mathbb{E}\{u_{t,k}\mid o_t\}f(1)(\tilde\delta_{t,k}^{\mathrm{f}})^2 + \mathbb{E}\{r_{t,k}\mid o_t\}f(x_{t,k}+1)(\tilde\delta_{t,k}^{\mathrm{p}})^2\\
    \hspace{1mm}+\big[1-f(1)\mathbb{E}\{u_{t,k}\mid o_t\}-f(x_{t,k}+1)\mathbb{E}\{r_{t,k}\mid o_t\}\big](\tilde\delta_{t,k} )^2\\ \hspace{1mm} + \sum_{k\in\mathcal{K}}\sum_{k'\in\mathcal{K},k'\neq k}\mathbb{E}\{u_{t,k}\mid o_t\}f(1)\tilde\delta_{t,k}^{\mathrm{f}}\tilde\delta_{t,k'} + \mathbb{E}\{r_{t,k}\mid o_t\}f(x_{t,k}+1)\tilde\delta_{t,k}^{\mathrm{p}} \tilde\delta_{t,k'} \\ \hspace{1mm}+ \mathbb{E}\{u_{t,k'}\mid o_t\}f(1)\tilde\delta_{t,k'}^{\mathrm{f}}\tilde\delta_{t,k} + \mathbb{E}\{r_{t,k'}\mid o_t\}f(x_{t,k'}+1)\tilde\delta_{t,k'}^{\mathrm{p}} \tilde\delta_{t,k} - \mathbb{E}\{u_{t,k}\mid o_t\} f(1)\tilde\delta_{t,k} \tilde\delta_{t,k'} \\\hspace{1mm}-\mathbb{E}\{r_{t,k}\mid o_t\} f(x_{t,k}+1)\tilde\delta_{t,k} \tilde\delta_{t,k'} -\mathbb{E}\{u_{t,k'}\mid o_t\} f(1)\tilde\delta_{t,k'} \tilde\delta_{t,k} \\
    \hspace{1mm}-\mathbb{E}\{r_{t,k'}\mid o_t\} f(x_{t,k'}+1)\tilde\delta_{t,k'} \tilde\delta_{t,k} +\tilde\delta_{t,k'} \tilde\delta_{t,k}\Big]+2Q_{t} \big(\frac{1}{K}\sum_{k\in\mathcal{K}}\mathbb{E}\{u_{t,k}\mid o_t\}f(1)\tilde\delta_{t,k}^{\mathrm{f}} \\
   \hspace{1mm}+ \mathbb{E}\{r_{t,k}\mid o_t\}f(x_{t,k}+1)\tilde\delta_{t,k}^{\mathrm{p}}\\
    \hspace{1mm}+ \big[1-f(1)\mathbb{E}\{u_{t,k}\mid o_t\}-f(x_{t,k}+1)\mathbb{E}\{r_{t,k}\mid o_t\}\big]\tilde\delta_{t,k} - {\Delta}^{\mathrm{max}} \big) .
\end{array}
\end{equation}
Dropping the expectation in \eqref{eqQa}, $W_t$ is derived as
\begin{equation*}\label{eqQ5}
\begin{array}{ll}
   W_t&\hspace{-3mm}= V\sum_{k\in \mathcal{K}}u_{t,k}+r_{t,k} +\frac{1}{2K^2}\Big[\sum_{k\in \mathcal{K}}u_{t,k}f(1) (\tilde\delta_{t,k}^{\mathrm{f}})^2  + r_{t,k}f(x_{t,k}+1)(\tilde\delta_{t,k}^{\mathrm{p}})^2\\&\hspace{1mm}+\big[ 1-u_{t,k} f(1) - r_{t,k} f(x_{t,k}+1) \big](\tilde\delta_{t,k})^2+\sum_{k\in \mathcal{K}}\sum_{k'\in \mathcal{K}, k'\neq k}u_{t,k}f(1)\tilde\delta_{t,k}^{\mathrm{f}}\tilde\delta_{t,k'}\\
   &\hspace{1mm}+ r_{t,k}f(x_{t,k}+1)\tilde\delta_{t,k}^{\mathrm{p}} \tilde\delta_{t,k'}- u_{t,k} f(1)\tilde\delta_{t,k} \tilde\delta_{t,k'}-r_{t,k} f(x_{t,k}+1)\tilde\delta_{t,k} \tilde\delta_{t,k'}\\
   &\hspace{1mm}+ u_{t,k'}f(1)\tilde\delta_{t,k'}^{\mathrm{f}}\tilde\delta_{t,k}+ r_{t,k'}f(x_{t,k'}+1)\tilde\delta_{t,k'}^{\mathrm{p}} \tilde\delta_{t,k} -u_{t,k'} f(1)\tilde\delta_{t,k'} \tilde\delta_{t,k} \\
   &\hspace{1mm}-r_{t,k'} f(x_{t,k'}+1)\tilde\delta_{t,k'} \tilde\delta_{t,k} +\tilde\delta_{t,k'} \tilde\delta_{t,k} +2KQ_{t} \big(\sum_{k\in \mathcal{K}}u_{t,k} f(1)\tilde\delta_{t,k}^{\mathrm{f}}+ r_{t,k}f(x_{t,k}+1)\tilde\delta_{t,k}^{\mathrm{p}} \\&\hspace{1mm}+\big[ 1-u_{t,k} f(1) - r_{t,k} f(x_{t,k}+1)    \big]\tilde\delta_{t,k}\Big ]+\frac{1}{2} \big(({\Delta}^{\mathrm{max}})^2 - 2Q_{t} {\Delta}^{\mathrm{max}}\big).
\end{array}
\end{equation*}




\ifCLASSOPTIONcaptionsoff
  \newpage
\fi



%
\bibliographystyle{IEEEtran}
\begin{spacing}{1.35}
\bibliography{conf_short,jour_short,main}

\begin{thebibliography}{10}
\providecommand{\url}[1]{#1}
\csname url@samestyle\endcsname
\providecommand{\newblock}{\relax}
\providecommand{\bibinfo}[2]{#2}
\providecommand{\BIBentrySTDinterwordspacing}{\spaceskip=0pt\relax}
\providecommand{\BIBentryALTinterwordstretchfactor}{4}
\providecommand{\BIBentryALTinterwordspacing}{\spaceskip=\fontdimen2\font plus
\BIBentryALTinterwordstretchfactor\fontdimen3\font minus
  \fontdimen4\font\relax}
\providecommand{\BIBforeignlanguage}[2]{{%
\expandafter\ifx\csname l@#1\endcsname\relax
\typeout{** WARNING: IEEEtran.bst: No hyphenation pattern has been}%
\typeout{** loaded for the language `#1'. Using the pattern for}%
\typeout{** the default language instead.}%
\else
\language=\csname l@#1\endcsname
\fi
#2}}
\providecommand{\BIBdecl}{\relax}
\BIBdecl

\bibitem{aoi1}
S.~Kaul, R.~Yates, and M.~Gruteser, ``Real-time status: How often should one
  update?'' in \emph{Proc. IEEE Int. Conf. on Computer Commun. (INFOCOM)},
  Orlando, FL, USA, Mar. 25--30, 2012, pp. 2731--2735.

\bibitem{oiot}
M.~A. Abd-Elmagid, N.~Pappas, and H.~S. Dhillon, ``On the role of age of
  information in the internet of things,'' \emph{{IEEE} Commun. Mag.}, vol.~57,
  no.~12, pp. 72--77, Dec. 2019.

\bibitem{aoi2}
S.~K. Kaul, R.~D. Yates, and M.~Gruteser, ``Status updates through queues,'' in
  \emph{Proc. Conf. Inform. Sciences Syst. (CISS)}, Princeton, NJ, USA,
  Mar.21--23, 2012, pp. 1--6.

\bibitem{ry}
R.~D. Yates, ``The age of information in networks: Moments, distributions, and
  sampling,'' \emph{{IEEE} Trans. Inform. Theory}, vol.~66, no.~9, pp.
  5712--5728, Sep. 2020.

\bibitem{pp}
A.~Kosta, N.~Pappas, and V.~Angelakis, ``Age of information: A new concept,
  metric, and tool,'' \emph{Found. Trends Netw.}, vol.~12, no.~3, pp. 162--259,
  Nov. 2017.

\bibitem{arqb}
S.~Lin, D.~J. Costello, and M.~J. Miller, ``Automatic-repeat-request
  error-control schemes,'' \emph{{IEEE} Commun. Mag.}, vol.~22, no.~12, pp.
  5--17, Dec. 1984.

\bibitem{harq}
``{IEEE} standard for air interface for broadband wireless access systems,''
  \emph{IEEE Std 802.16-2017 (Revision of IEEE Std 802.16-2012)}, pp. 1--2726,
  Mar. 2018.

\bibitem{eiage}
V.~Raghunathan, C.~Schurgers, S.~Park, and M.~Srivastava, ``Energy-aware
  wireless microsensor networks,'' \emph{{IEEE} Signal Processing Mag.},
  vol.~19, no.~2, pp. 40--50, Mar. 2002.

\bibitem{env}
S.~Russell and P.~Norvig, \emph{Artificial intelligence: a modern approach},
  2002.

\bibitem{beutler}
F.~J. Beutler and K.~W. Ross, ``Optimal policies for controlled {Markov} chains
  with a constraint,'' \emph{Journal of mathematical analysis and
  applications}, vol. 112, no.~1, pp. 236--252, Nov. 1985.

\bibitem{bz}
B.~Zhou and W.~Saad, ``Joint status sampling and updating for minimizing age of
  information in the internet of things,'' \emph{{IEEE} Trans. Commun.},
  vol.~67, no.~11, pp. 7468--7482, Nov. 2019.

\bibitem{ngz}
E.~T. Ceran, D.~Gündüz, and A.~György, ``A reinforcement learning approach
  to age of information in multi-user networks with {HARQ},'' \emph{{IEEE} J.
  Select. Areas Commun.}, vol.~39, no.~5, pp. 1412--1426, May 2021.

\bibitem{gz}
E.~T. Ceran, D.~G{\"u}nd{\"u}z, and A.~Gy{\"o}rgy, ``Average age of information
  with hybrid {ARQ} under a resource constraint,'' \emph{{IEEE} Trans. Wireless
  Commun.}, vol.~18, no.~3, pp. 1900--1913, Mar. 2019.

\bibitem{lyp}
M.~J. Neely, \emph{Stochastic network optimization with application to
  communication and queueing systems}.\hskip 1em plus 0.5em minus 0.4em\relax
  Morgan \& Claypool Publishers, 2010.

\bibitem{dqn}
V.~Mnih \emph{et~al.}, ``Human-level control through deep reinforcement
  learning,'' \emph{Nature}, vol. 518, no. 7540, pp. 529--533, Feb. 2015.

\bibitem{ry2}
R.~D. Yates and S.~K. Kaul, ``The age of information: Real-time status updating
  by multiple sources,'' \emph{{IEEE} Trans. Inform. Theory}, vol.~65, no.~3,
  pp. 1807--1827, Mar. 2018.

\bibitem{mm1}
M.~Moltafet, M.~Leinonen, and M.~Codreanu, ``Average {AoI} in multi-source
  systems with source-aware packet management,'' \emph{{IEEE} Trans. Commun.},
  vol.~69, no.~2, pp. 1121--1133, Feb. 2021.

\bibitem{enj2}
E.~Najm, R.~Yates, and E.~Soljanin, ``Status updates through {M/G/1/1} queues
  with {HARQ},'' in \emph{Proc. IEEE Int. Symp. Inform. Theory}, Aachen,
  Germany, Jun. 25--30, 2017, pp. 131--135.

\bibitem{nr}
D.~Li, S.~Wu, Y.~Wang, J.~Jiao, and Q.~Zhang, ``Age-optimal {HARQ} design for
  freshness-critical satellite-{IoT} systems,'' \emph{{IEEE} Internet of Things
  J.}, vol.~7, no.~3, pp. 2066--2076, Mar. 2020.

\bibitem{nnr}
S.~Liu, J.~Jiao, Z.~Ni, S.~Wu, and Q.~Zhang, ``Age-optimal {NC-HARQ} protocol
  for multi-hop satellite-based internet of things,'' in \emph{Proc. IEEE
  Wireless Commun. and Networking Conf.}, Nanjing, China, Mar. 29--1, 2021, pp.
  1--6.

\bibitem{frzi}
S.~Farazi, A.~G. Klein, and D.~R. Brown, ``Average age of information in update
  systems with active sources and packet delivery errors,'' \emph{{IEEE}
  Wireless Commun. Lett.}, vol.~9, no.~8, pp. 1164--1168, Aug. 2020.

\bibitem{shi2021}
Y.~Shi, L.~Jing, X.~Jia, P.~Ji, and N.~Wan, ``Improvement on age of information
  for information update systems with {HARQ} chase combining and sensor
  harvesting-transmitting diversities,'' \emph{{IEEE} Access}, vol.~9, pp.
  78\,035--78\,049, May 2021.

\bibitem{arafa2021}
A.~Arafa, J.~Yang, S.~Ulukus, and H.~V. Poor, ``Timely status updating over
  erasure channels using an energy harvesting sensor: Single and multiple
  sources,'' \emph{{IEEE} Trans. Green Commun. Net.}, vol.~6, no.~1, pp. 6--19,
  Mar. 2022.

\bibitem{deng2021}
Z.~Deng, S.~Wu, C.~Guo, J.~Jiao, N.~Zhang, and Q.~Zhang, ``Age-optimal
  transmission policy for intelligent {HARQ-CC} aided {NOMA} systems,'' in
  \emph{Proc. IEEE Int. Conf. Commun.}, Montreal, QC, Canada, Jun. 14--23,
  2021, pp. 1--6.

\bibitem{feng2021}
S.~Feng and J.~Yang, ``Age of information minimization for an energy harvesting
  source with updating erasures: Without and with feedback,'' \emph{{IEEE}
  Trans. Commun.}, vol.~69, no.~8, pp. 5091--5105, May 2021.

\bibitem{wang2020}
Y.~Wang, S.~Wu, J.~Jiao, W.~Wu, Y.~Wang, and Q.~Zhang, ``Age-optimal
  transmission policy with {HARQ} for freshness-critical vehicular status
  updates in space-air-ground integrated networks,'' \emph{{IEEE} Internet of
  Things J.}, vol.~9, no.~8, pp. 5719--5729, Apr. 2022.

\bibitem{frr}
X.~Lagrange, ``Throughput of {HARQ} protocols on a block fading channel,''
  \emph{{IEEE} Commun. Lett.}, vol.~14, no.~3, pp. 257--259, Mar. 2010.

\bibitem{harqc}
P.~Frenger, S.~Parkvall, and E.~Dahlman, ``Performance comparison of {HARQ}
  with chase combining and incremental redundancy for {HSDPA},'' in \emph{Proc.
  IEEE Veh. Technol. Conf.}, Atlantic City, NJ, USA, Oct. 2001, pp. 1829--1833.

\bibitem{mmpower}
M.~Moltafet, M.~Leinonen, M.~Codreanu, and N.~Pappas, ``Power minimization for
  age of information constrained dynamic control in wireless sensor networks,''
  \emph{{IEEE} Trans. Commun.}, vol.~70, no.~1, pp. 419--432, Jan. 2022.

\bibitem{9241401}
A.~Maatouk, S.~Kriouile, M.~Assad, and A.~Ephremides, ``On the optimality of
  the {Whittle’s} index policy for minimizing the age of information,''
  \emph{{IEEE} Trans. Wireless Commun.}, vol.~20, no.~2, pp. 1263--1277, Oct.
  2021.

\bibitem{shr}
G.~Yao, A.~Bedewy, and N.~B. Shroff, ``Age-optimal low-power status update over
  time-correlated fading channel,'' \emph{{IEEE} Trans. Mobile Comput.}, pp.
  1--1, Mar. 2022.

\bibitem{kr}
D.~V. Djonin and V.~Krishnamurthy, ``Mimo transmission control in fading
  channels—a constrained markov decision process formulation with monotone
  randomized policies,'' \emph{{IEEE} Trans. Signal Processing}, vol.~55,
  no.~10, pp. 5069--5083, Oct. 2007.

\bibitem{elman}
E.~Altman, \emph{Constrained {Markov} decision processes}.\hskip 1em plus 0.5em
  minus 0.4em\relax CRC Press, 1999.

\bibitem{nt}
D.-j. Ma, A.~M. Makowski, and A.~Shwartz, ``Estimation and optimal control for
  constrained {Markov} chains,'' in \emph{1986 25th IEEE Conf. on Decision and
  Contr.}, Athens, Greece, Dec. 1986, pp. 994--999.

\bibitem{puterman1994}
M.~L. Puterman, \emph{{Markov} Decision Processes: Discrete Stochastic Dynamic
  Programming}.\hskip 1em plus 0.5em minus 0.4em\relax John Wiley \& Sons,
  Inc., 1994.

\bibitem{wu2020}
W.~Wu, P.~Yang, W.~Zhang, C.~Zhou, and X.~Shen, ``Accuracy-guaranteed
  collaborative {DNN} inference in industrial {IoT} via deep reinforcement
  learning,'' \emph{{IEEE} Trans. Ind. Informat.}, vol.~17, no.~7, pp.
  4988--4998, Aug. 2020.

\end{thebibliography}
\end{spacing}
\end{document}